\pdfoutput=1
\documentclass[
    affil-it, % affiliations in italic
    auth-lg, % authors in large
    british, % dates in British format
    compression, % compressed title page
    contents, % show table of contents
    custom-packages={}, % Load listed custom packages
    references={bigone}, % Name of bibliography
]{lib/preprint}

\let\oldblacksquare\blacksquare
\newcommand{\BBox}{{\textcolor{gray}{\mathop\oldblacksquare\nolimits}}} % make less noticeable

\begin{document}
    
    % Hypersetup
    \hypersetup{
        pdftitle = {Ambitwistor Yang-Mills theory revisted},
        pdfauthor = {Leron Borsten,Branislav Jurco,Hyungrok Kim,Christian Saemann,Martin Wolf},
        pdfkeywords = {ambitwistor space,Yang-Mills theory,Chern-Simons theory,CR structure,supersymmetry}
    }
    
    % Date of preprint
    \date{\today}
    
    % All emails in the order of appearance
    \email{l.borsten@herts.ac.uk,h.kim2@herts.ac.uk,branislav.jurco@gmail.com,c.saemann@hw.ac.uk,m.wolf@surrey.ac.uk}
    
    % All preprint numbers in the order of appearance
    \preprint{DMUS--MP--24/06}
    
    % Title
    \title{Ambitwistor Yang--Mills Theory Revisited} 
    
    % All authors
    \author[a]{Leron~Borsten\,\orcidlink{0000-0001-9008-7725}\,}
    \author[b]{Branislav~Jur{\v c}o\,\orcidlink{0000-0001-7782-2326}\,}
    \author[a]{Hyungrok~Kim\,\orcidlink{0000-0001-7909-4510}\,}
    \author[c]{Christian~Saemann\,\orcidlink{0000-0002-5273-3359}\,}
    \author[d]{Martin~Wolf\,\orcidlink{0009-0002-8192-3124}\,}
    
    % All affiliations
    \affil[a]{Department of Physics, Astronomy, and Mathematics,\\University of Hertfordshire, Hatfield AL10 9AB, United Kingdom}
    \affil[b]{Mathematical Institute, Faculty of Mathematics and Physics,\\ Charles University, Prague 186 75, Czech Republic}
    \affil[c]{Maxwell Institute for Mathematical Sciences, Department of Mathematics,\\ Heriot--Watt University, Edinburgh EH14 4AS, United Kingdom}
    \affil[d]{School of Mathematics and Physics,\\ University of Surrey, Guildford GU2 7XH, United Kingdom}
    
    % Abstract
    \abstract{Inspired by the Movshev--Mason--Skinner Cauchy--Riemann (CR) ambitwistor approach, we provide a rigorous yet elementary construction of a twisted CR holomorphic Chern--Simons action on CR ambitwistor space for maximally supersymmetric Yang--Mills theory on four-dimensional Euclidean space. The key ingredient in our discussion is the homotopy algebraic perspective on perturbative quantum field theory. Using this technology, we show that both theories are semi-classically equivalent, that is, we construct a quasi-isomorphism between the cyclic $L_\infty$-algebras governing both field theories. This confirms a conjecture from the literature. Furthermore, we also show that the Yang--Mills action is obtained by integrating out an infinite tower of auxiliary fields in the Chern--Simons action, that is, the two theories are related by homotopy transfer. Given its simplicity, this Chern--Simons action should form a fruitful starting point for analysing perturbative properties of Yang--Mills theory.}
    
    % Acknowledgements
    \acknowledgements{We thank Tim Adamo for conversations.}
    
    % Declarations
    \declarations{
        \textbf{Funding.}
        B.~J.~was supported by the GA{\v C}R Grants EXPRO 19-28628X and 24-10887S. H.~K.~was partially supported by the Leverhulme Research Project Grant RPG-2021-092.\\[5pt]
        \textbf{Conflict of interest.}
        The authors have no relevant financial or non-financial interests to disclose.\\[5pt]
        \textbf{Data statement.}
        No additional research data beyond the data presented and cited in this work are needed to validate the research findings in this work.\\[5pt]
        \textbf{Licence statement.}
        For the purpose of open access, the authors have applied a Creative Commons Attribution (CC-BY) license to any author-accepted manuscript version arising.
    }
    
    % Body
    \begin{body}
        
        \section{Introduction and summary}
        
        It is well known that holomorphic Chern--Simons theory on the ambitwistor space is classically equivalent to $\caN=3$ supersymmetric Yang--Mills theory in four dimensions (a theory perturbatively equivalent to $\caN=4$ supersymmetric Yang--Mills theory) at the level of the moduli spaces of solutions and their gauge equivalence classes~\cite{Witten:1978xx,Isenberg:1978kk,Khenkin:1980ff,Harnad:1984vk,Harnad:1985bc,Buchdahl:1985aa,Eastwood:1987aa,Manin:1988ds,Nair:1988bq} (see also~\cite{Popov:2004rb} for a review). The construction of a corresponding action functional, however, is non-trivial. Whilst ambitwistor space is a Calabi--Yau supermanifold, and hence a natural holomorphic volume form exists~\cite{Witten:2003nn}, the required Lagrangian is unclear. In particular, ambitwistor space has a five-dimensional body, which is incompatible with the usual Chern--Simons-type Lagrangian.\footnote{It is, however, possible to construct an action for `higher' holomorphic Chern--Simons theory with a Lie 3-group as gauge group, providing a higher ambitwistor space action functional for maximally supersymmetric Yang--Mills theory~\cite{Samann:2017dah}. See~\cite{Borsten:2024gox} for a recent review on higher gauge theory.}
        
        For non-supersymmetric Yang--Mills theory in Euclidean signature, this obstruction can be evaded by using the Movshev--Mason--Skinner Cauchy--Riemann (CR) ambitwistor construction~\cite{Movshev:2004ub,Mason:2005kn}. Whilst the treatment in~\cite{Movshev:2004ub,Mason:2005kn} is rigorous and complete in the non-supersymmetric case, the generalisation to the supersymmetric setting is rather indirect and somewhat conjectural, principally due to the subtleties entailed by the superspace torsion and complex fermions in Euclidean signature. Here, we resolve these issues by providing a rigorous yet elementary construction of a twisted CR holomorphic Chern--Simons action for Yang--Mills theory in the case of $\caN=3$ supersymmetry. See e.g.~\cite{Mason:2005zm,Popov:2005uv,Boels:2006ir,Popov:2009nx,Popov:2021mfl} for other approaches of formulating twistorial actions for Yang--Mills theories with varying amounts of supersymmetry.\footnote{See~\cite{Berkovits:0406051,Wolf:2007tx,Mason:2007ct,Mason:2008jy,Adamo:2013tja} for twistorial action formulations of supergravity theories.} 
        
        Key to our discussion is the homotopy algebraic perspective on perturbative quantum field theory. This is a dictionary between objects and tools in perturbative quantum field theory on the one side and objects and tools in homotopical algebra on the other side. In particular, any perturbative field theory given in terms of a polynomial action can be reformulated as a cyclic $L_\infty$-algebra, cf.~e.g.~\cite{Jurco:2018sby}\footnote{or~\cite{Hohm:2017pnh} at the level of equations of motion}. Similarly, semi-classical equivalence\footnote{i.e.~having isomorphic tree-level S-matrices} corresponds to the natural notion of equivalence for these $L_\infty$-algebras, namely quasi-isomorphisms. The final aspect of this dictionary we will put to use is that the procedure of integrating out fields is encoded in the so-called homotopy transfer, a technique from homological perturbation theory which reproduces the tree-level Feynman-diagram expansion~\cite{Kajiura:2003ax,Doubek:2017naz,Jurco:2018sby,Macrelli:2019afx,Arvanitakis:2019ald,Jurco:2019yfd,Saemann:2020oyz}.
        
        Using this technology, we prove that the aforementioned twisted CR holomorphic Chern--Simons theory is, in fact, semi-classically equivalent to $\caN=3$ supersymmetric Yang--Mills theory. This is the statement of \cref{thm:equivalence}, our first central result. This confirms a conjecture made in~\cite{Movshev:2004ub}, albeit using an $L_\infty$-algebra that is differs in some key details from the one we use here. In \cref{thm:homotopy-transfer}, our second central result, we moreover show that this semi-classical equivalence simply amounts to integrating out an infinite tower of auxiliary fields, that is, to homotopy transfer. 
        
        Given its simplicity, our CR ambitwistor action should form a fruitful starting point for analysing various aspects of Yang--Mills theory; see the reviews~\cite{Woodhouse:id,Wolf:2010av,Adamo:2011pv,Adamo:2013cra,Adamo:2017qyl} and the references therein for twistor theory actions and their applications.\footnote{See~\cite{Wolf:2004hp,Wolf:2005sd,Popov:2006qu,Wolf:2006me,Popov:2021nfh} for analysing classical integrability of maximally supersymmetric Yang--Mills theory using twistor methods.} Indeed, the primary motivation underlying the present contribution is the conjectured colour--kinematics duality of Yang--Mills scattering amplitudes~\cite{Bern:2008qj,Bern:2010ue,Bern:2010yg} (see~\cite{Bern:2019prr,Borsten:2020bgv} for reviews). In~\cite{Borsten:2022vtg}, we presented a route to proving all-loop-order colour--kinematics duality for maximally supersymmetric Yang--Mills theory. The key observation is that the tower of interaction terms constructed in~\cite{Borsten:2020zgj,Borsten:2021hua,Borsten:2021gyl} to generate colour--kinematics-duality-respecting Feynman rules may be geometrised\footnote{whilst also avoiding the potential need for non-local field redefinitions that require counterterms breaking colour--kinematics duality, as discussed in detail in~\cite{Borsten:2021gyl}} as the `Kaluza--Klein tower' following from dimensional reduction on the $\IC P^1\times\IC P^1$ of the ambitwistor space.\footnote{A conceptually similar approach employing, instead, pure spinor space has been used to show that ten-dimensional supersymmetric Yang--Mills theory and its dimensional reductions have tree-level colour--kinematics duality~\cite{Ben-Shahar:2021doh,Borsten:2023reb} and, similarly, that the low energy effective world-volume theories of M2-branes have tree-level colour--kinematics duality~\cite{Borsten:2023reb, Borsten:2023ned}, as previously conjectured in~\cite{Bargheer:2012gv,Huang:2013kca}. The restriction to the tree-level in both cases is a consequence of the need to regularise the integrals over pure spinor space, i.e.~there is no known scheme that is guaranteed to preserve colour--kinematics duality at the loop-level. In this regard the use of twistor spaces has a clear advantage.} In doing so, the proof of colour--kinematics duality is simplified to establishing that our CR ambitwistor action for Yang--Mills theory has an associated BV${}^{\color{gray}\blacksquare}$-algebra\footnote{See~\cite{Reiterer:2019dys} for the original definition of BV${}^{\color{gray}\blacksquare}$-algebras, where they were used to prove colour--kinematics duality for the tree-level Yang--Mills S-matrix. See~\cite{Borsten:2021gyl,Ben-Shahar:2021doh,Ben-Shahar:2021zww,Borsten:2022ouu,Bonezzi:2022bse,Bonezzi:2023ced,Bonezzi:2023ciu,Bonezzi:2023lkx,Bonezzi:2023pox,Borsten:2023ned,Borsten:2023paw,Bonezzi:2024dlv,Armstrong-Williams:2024icu} for related work on BV${}^{\color{gray}\blacksquare}$-algebras, colour--kinematics duality, and the double copy at the level of actions.}, together with an appropriate Ward-identity-like argument for the twistorial gauge symmetries~\cite{Borsten:2022vtg}. Evidently, a prerequisite for this to work is the generalisation of the semi-classical equivalence to the loop level and this, in turn, would most certainly rely on the vanishing of certain twistor anomalies~\cite{Costello:2021bah}. This approach to colour--kinematics duality has successfully been applied to self-dual supersymmetric Yang--Mills theory~\cite{Borsten:2022vtg,Borsten:2023paw,Borsten:2023ned}, where the use of holomorphic Chern--Simons theory and Penrose's twistor space obviates the need to apply a supplementary Ward identity argument. Hence, the generalisation of this approach to full maximally supersymmetric Yang--Mills theory would first require a complete, rigorous, and explicit formulation of a Chern--Simons-type action, which we provide here. 
        
        \section{CR ambitwistor space and supersymmetric Yang--Mills equations}
        
        We shall use the conventions of the review~\cite{Wolf:2010av} in the following. See also~\cite{Woodhouse:id,Adamo:2011pv,Adamo:2013cra,Adamo:2017qyl} for other reviews on twistor theory and applications to gauge theory.
        
        \subsection{Bosonic CR ambitwistor space}
        
        \paragraph{CR ambitwistor space.}
        Given the factorisation $T_\IC\IR^4\cong S\otimes\tilde S$ of the complexified tangent bundle of $\IR^4$ into $S$, the chiral spin bundle, and $\tilde S$, the anti-chiral spin bundle, we consider the projectivisation $\IP(S^*)\times\IP(\tilde S^*)\rightarrow\IR^4$. Evidently, this space is
        \begin{equation}\label{eq:CRAmbitwistorSpace}
            L\ \coloneqq\ \IR^4\times\IC P^1\times\IC P^1~.
        \end{equation}
        We may choose $(x^{\alpha\dot\alpha},\lambda_{\dot\alpha},\mu_\alpha)$ as coordinates on $L$ for $\alpha,\beta,\ldots,\dot\alpha,\dot\beta,\ldots=1,2$ with $x^{\alpha\dot\alpha}$ coordinates\footnote{We have $x^{\alpha\dot\alpha}=\sigma^{\alpha\dot\alpha}_\mu x^\mu$ with $\sigma^{\alpha\dot\alpha}_\mu$ the standard sigma matrices.} on Euclidean $\IR^4$ and $(\lambda_{\dot\alpha},\mu_\alpha)$ complex homogeneous coordinates on $\IC P^1\times\IC P^1$. The coordinates manifest the action of $\sfSpin(4)\cong \sfSU(2)\times \sfSU(2)$ on $L$.
        
        On the fibres of $L\rightarrow\IR^4$ there is a natural \uline{quaternionic structure}\footnote{Recall that a quaternionic structure is an anti-linear endomorphism $\tau$ with $\tau^2=-\sfid$.}
        \begin{equation}\label{eq:quaternionicStructure}
            \tau\,:\,(\lambda_{\dot\alpha},\mu_\alpha)\ =\ 
            \left(
            \begin{pmatrix}
                \lambda_{\dot1}\\ \lambda_{\dot2}
            \end{pmatrix},
            \begin{pmatrix}
                \mu_1\\ \mu_2
            \end{pmatrix}
            \right)
            \ \mapsto\ 
            (\hat\lambda_{\dot\alpha},\hat\mu_\alpha)\ \coloneqq\ 
            \left(
            \begin{pmatrix}
                -\overline{\lambda_{\dot2}}\\ \overline{\lambda_{\dot1}}
            \end{pmatrix},
            \begin{pmatrix}
                -\overline{\mu_2}\\ \overline{\mu_1}
            \end{pmatrix}
            \right),
        \end{equation}
        where the bar denotes complex conjugation. This induces the spinor norms
        \begin{equation}\label{eq:spinorNorms}
            |\lambda|^2\ \coloneqq\ \lambda_{\dot\alpha}\hat\lambda^{\dot\alpha}
            \ewith
            \hat\lambda^{\dot\alpha}\ =\ \eps^{\dot\alpha\dot\beta}\hat\lambda_{\dot\beta}
            \eand
            |\mu|^2\ \coloneqq\ \mu_\alpha\hat\mu^\alpha
            \ewith
            \hat\mu^\alpha\ =\ \eps^{\alpha\beta}\hat\mu_\beta~.
        \end{equation}
        Here, $\eps_{\alpha\beta}=-\eps_{\beta\alpha}$ and $\eps_{\dot\alpha\dot\beta}=-\eps_{\dot\beta\dot\alpha}$ are the standard symplectic structures on $S$ and $\tilde S$ with $\eps_{12}=1=\eps_{\dot1\dot2}$. We also define $\eps^{\alpha\beta}$ and $\eps^{\dot\alpha\dot\beta}$ by $\eps_{\alpha\gamma}\eps^{\gamma\beta}=\delta_\alpha{}^\beta$ and $\eps_{\dot\alpha\dot\gamma}\eps^{\dot\gamma\dot\beta}=\delta_{\dot\alpha}{}^{\dot\beta}$, respectively, and raise and lower spinorial indices using these symplectic forms. Note that $\tau(|\lambda|^2)=|\lambda|^2$ and $\tau(|\mu|^2)=|\mu|^2$. Furthermore, $\tau$ also induces the reality conditions
        \begin{equation}\label{eq:EuclideanReality}
            \overline{x^{1\dot1}}\ =\ x^{2\dot 2}
            \eand
            \overline{x^{1\dot 2}}\ =\ -x^{2\dot 1}~,
        \end{equation}
        so that the Euclidean line element is simply
        \begin{equation}\label{eq:metricR4}
            (\rmd s_{\IR^4})^2\ =\ \tfrac12\eps_{\alpha\beta}\eps_{\dot\alpha\dot\beta}\rmd x^{\alpha\dot\alpha}\rmd x^{\beta\dot\beta}\ =\ |\rmd x^{1\dot1}|^2+|\rmd x^{1\dot2}|^2~.
        \end{equation}
        
        We note that $L$ can be shown to be diffeomorphic to a real quadric of co-dimension two inside the complex ambitwistor space (which, in turn, is a complex quadric hypersurface in the product of the Penrose twistor space with its dual)~\cite{Movshev:2004ub,Mason:2005kn}. Following these papers, we call~\eqref{eq:CRAmbitwistorSpace} the \uline{CR ambitwistor space}. 
        
        \paragraph{CR structure.}
        The reason for this name stems from the fact that $L$ comes with a natural almost \uline{CR structure}\footnote{See e.g.~the text book~\cite{Boggess:1991aa} for a general account on CR manifolds.}
        \begin{subequations}\label{eq:tangentBundleDistributionsCRAmbitwistorSpace}
            \begin{equation}
                \begin{gathered}
                    T^{0,1}_{\rm CR}L\ \coloneqq\ {\rm span}\{\hat E_\rmF,\hat E_\rmL,\hat E_\rmR\}~,
                    \\
                    \hat E_\rmF\ \coloneqq\ \mu^\alpha\lambda^{\dot\alpha}\parder{x^{\alpha\dot\alpha}}~,~~~
                    \hat E_\rmL\ \coloneqq\ |\lambda|^2\lambda_{\dot\alpha}\parder{\hat\lambda_{\dot\alpha}}~,~~~
                    \hat E_\rmR\ \coloneqq\ |\mu|^2\mu_\alpha\parder{\hat\mu_\alpha}~,
                \end{gathered}
            \end{equation}
            where $|\lambda|^2$ and $|\mu|^2$ were defined in~\eqref{eq:spinorNorms}. The use of `$\rmF$', `$\rmL$', and `$\rmR$' is to remind the reader that the corresponding vector fields are along the flat Euclidean, left $\IC P^1$, and right $\IC P^1$ directions, respectively. We also set
            \begin{equation}
                \begin{gathered}
                    \overline{T^{0,1}_{\rm CR}L}\ \coloneqq\ {\rm span}\{E_\rmF,E_\rmL,E_\rmR\}~,
                    \\
                    E_\rmF\ \coloneqq\ \frac{\hat\mu^\alpha\hat\lambda^{\dot\alpha}}{|\mu|^2|\lambda|^2}\parder{x^{\alpha\dot\alpha}}~,~~~
                    E_\rmL\ \coloneqq\ -\frac{\hat\lambda_{\dot\alpha}}{|\lambda|^2}\parder{\lambda_{\dot\alpha}}~,~~~
                    E_\rmR\ \coloneqq\ -\frac{\hat\mu_\alpha}{|\mu|^2}\parder{\mu_\alpha}~,
                \end{gathered}
            \end{equation}
            and
            \begin{equation}
                W\ \coloneqq\ {\rm span}\{E_\rmW,E_{\hat\rmW}\}
                \ewith
                E_\rmW\ \coloneqq\ \frac{\mu^\alpha\hat\lambda^{\dot\alpha}}{|\lambda|^2}\parder{x^{\alpha\dot\alpha}}
                \eand
                E_{\hat\rmW}\ \coloneqq\ -\frac{\hat\mu^\alpha\lambda^{\dot\alpha}}{|\mu|^2}\parder{x^{\alpha\dot\alpha}}~.
            \end{equation}
        \end{subequations}
        Evidently, each of these distributions is integrable. We also have the decomposition
        \begin{equation}\label{eq:decompositionTangentBundleCRAmbitwistorSpace}
            T_\IC L\ \cong\ T^{0,1}_{\rm CR}L\oplus\underbrace{\overline{T^{0,1}_{\rm CR}L}\oplus W}_{\eqqcolon\,T^{1,0}_{\rm CR}L}~.
        \end{equation}
        Moreover, we have $\tau(T^{0,1}_{\rm CR}L)=\overline{T^{0,1}_{\rm CR}L}$ and $\tau(W)=W$ under the quaternionic structure defined in~\eqref{eq:quaternionicStructure}, that is, the distributions $\overline{T^{0,1}_{\rm CR}L}$ and $T^{0,1}_{\rm CR}L$ are complex conjugates of each other and $W$ is real. Note also that $T^{0,1}_{\rm CR}L\cap\overline{T^{0,1}_{\rm CR}L}=\{0\}$ as is required for an (almost) CR structure.
        
        Dually, we have
        \begin{subequations}\label{eq:coTangentBundleDistributionsCRAmbitwistorSpace}
            \begin{equation}
                \begin{gathered}
                    \Omega^{0,1}_{\rm CR}(L)\ \coloneqq\ {\rm span}\{\hat e^\rmF,\hat e^\rmL,\hat e^\rmR\}~,
                    \\
                    \hat e^\rmF\ \coloneqq\ \frac{\rmd x^{\alpha\dot\alpha}\hat\mu_\alpha\hat\lambda_{\dot\alpha}}{|\mu|^2|\lambda|^2}~,~~~
                    \hat e^\rmL\ \coloneqq\ \frac{\rmd\hat\lambda_{\dot\alpha}\hat\lambda^{\dot\alpha}}{|\lambda|^4}~,~~~
                    \hat e^\rmR\ \coloneqq\ \frac{\rmd\hat\mu_\alpha\hat\mu^\alpha}{|\mu|^4}~,
                \end{gathered}
            \end{equation}
            as well as
            \begin{equation}
                \begin{gathered}
                    \overline{\Omega^{0,1}_{\rm CR}(L)}\ \coloneqq\ {\rm span}\{e^\rmF,e^\rmL,e^\rmR\}~,
                    \\
                    e^\rmF\ \coloneqq\ \rmd x^{\alpha\dot\alpha}\mu_\alpha\lambda_{\dot\alpha}~,~~~
                    e^\rmL\ \coloneqq\ \rmd\lambda_{\dot\alpha}\lambda^{\dot\alpha}~,~~~
                    e^\rmR\ \coloneqq\ \rmd\mu_\alpha\mu^\alpha~,
                \end{gathered}
            \end{equation}
            and
            \begin{equation}
                \Omega^1_W(L)\ \coloneqq\ {\rm span}\{e^\rmW,e^{\hat\rmW}\}
                \ewith
                e^\rmW\ \coloneqq\ -\frac{\rmd x^{\alpha\dot\alpha}\hat\mu_\alpha\lambda_{\dot\alpha}}{|\mu|^2}
                \eand
                e^{\hat\rmW}\ \coloneqq\ \frac{\rmd x^{\alpha\dot\alpha}\mu_\alpha\hat\lambda_{\dot\alpha}}{|\lambda|^2}~.
            \end{equation}
        \end{subequations}
        Furthermore, because of~\eqref{eq:decompositionTangentBundleCRAmbitwistorSpace}, the exterior derivative on $L$ decomposes as $\rmd=\bar\partial_{\rm CR}+\partial_{\rm CR}$ with
        \begin{equation}\label{eq:CRAmbitwistorSpaceDifferentials}
            \begin{aligned}
                \bar\partial_{\rm CR}\ &\coloneqq\ \hat e^\rmF\hat E_\rmF+\hat e^\rmL\hat E_\rmL+\hat e^\rmR\hat E_\rmR~,
                \\
                \partial_{\rm CR}\ &\coloneqq\ e^\rmF E_\rmF+e^\rmL E_\rmL+e^\rmR E_\rmR+e^\rmW E_\rmW+e^{\hat\rmW}E_{\hat\rmW}~.
            \end{aligned}
        \end{equation}
        The commutation relations amongst the vector fields~\eqref{eq:tangentBundleDistributionsCRAmbitwistorSpace} then yield 
        \begin{equation}
            \begin{gathered}
                \bar\partial_{\rm CR}e^\rmW\ =\ e^\rmL\wedge\hat e^\rmF-\hat e^\rmR\wedge e^\rmF~,~~~
                \bar\partial_{\rm CR}e^{\hat\rmW}\ =\ \hat e^\rmL\wedge e^\rmF-e^\rmR\wedge\hat e^\rmF~,
                \\
                \partial_{\rm CR}e^\rmF\ =\ -e^\rmL\wedge e^{\hat\rmW}+e^\rmR\wedge e^\rmW~,~~~
                \partial_{\rm CR}\hat e^\rmF\ =\ -\hat e^\rmL\wedge e^\rmW+\hat e^\rmR\wedge e^{\hat\rmW}
            \end{gathered}
        \end{equation}
        as the only non-vanishing actions of $\bar\partial_{\rm CR}$ and $\partial_{\rm CR}$, respectively.
        
        \subsection{Supersymmetric extension and constraint system}
        
        \paragraph{Euclidean $\caN=3$ supersymmetry.}
        To incorporate Euclidean $\caN=3$ supersymmetry, we enlarge~\eqref{eq:CRAmbitwistorSpace} and consider
        \begin{equation}
            F\ \coloneqq\ \IR^{4|12}_{\rm cpl}\times\IC P^1\times\IC P^1
            \ewith
            \IR^{4|12}_{\rm cpl}\ \coloneqq\ \IR^{4|0}\times\IC^{0|12}~.
        \end{equation} 
        We coordinatise $\IR^{4|12}_{\rm cpl}$ by $(x^{\alpha\dot\alpha},\eta^{\dot\alpha}_i,\theta^{i\alpha})$ with $i,j,\ldots=1,2,3$ and $\IC P^1\times\IC P^1$ by $(\lambda_{\dot\alpha},\mu_\alpha)$, respectively. In this Euclidean setting, the fermionic coordinates remain complex\footnote{This is a necessity, as the only reality condition for Euclidean spinors is a symplectic Majorana condition, requiring an even amount of supersymmetry.}, and we restrict to functions and differential forms on $F$ that depend on the fermionic coordinates only via $\eta^{\dot\alpha}_i$ and $\theta^{i\alpha}$ and \uline{not} their complex conjugates. We shall refer to $F$ as the \uline{augmented CR ambitwistor space}. The (integrable) CR structure~\eqref{eq:tangentBundleDistributionsCRAmbitwistorSpace} extends as
        \begin{subequations}
            \begin{equation}\label{eq:tangentBundleDistributionsCRSuperAmbitwistorSpace}
                \begin{gathered}
                    T^{0,1}_{\rm CR}F\ \coloneqq\ {\rm span}\{\hat E_\rmF,\hat E_\rmL,\hat E_\rmR,\hat E^i,\hat E_i\}~,
                    \\
                    \hat E_\rmF\ \coloneqq\ \mu^\alpha\lambda^{\dot\alpha}\parder{x^{\alpha\dot\alpha}}~,~~~
                    \hat E_\rmL\ \coloneqq\ |\lambda|^2\lambda_{\dot\alpha}\parder{\hat\lambda_{\dot\alpha}}~,~~~
                    \hat E_\rmR\ \coloneqq\ |\mu|^2\mu_\alpha\parder{\hat\mu_\alpha}~,
                    \\
                    \hat E^i\ \coloneqq\ \lambda^{\dot\alpha}\underbrace{\left(\parder{\eta^{\dot\alpha}_i}+\theta^{i\alpha}\parder{x^{\alpha\dot\alpha}}\right)}_{\eqqcolon\,D^i_{\dot\alpha}},~~~
                    \hat E_i\ \coloneqq\ \mu^\alpha\underbrace{\left(\parder{\theta^{i\alpha}}+\eta^{\dot\alpha}_i\parder{x^{\alpha\dot\alpha}}\right)}_{\eqqcolon\,D_{i\alpha}}~,
                \end{gathered}
            \end{equation}
            and, dually,
            \begin{equation}
                \begin{gathered}
                    \Omega^{0,1}_{\rm CR}(F)\ \coloneqq\ {\rm span}\{\hat e^\rmF,\hat e^\rmL,\hat e^\rmR,\hat e_i,\hat e^i\}~,
                    \\
                    \hat e^\rmF\ \coloneqq\ \frac{(\rmd x^{\alpha\dot\alpha}+\theta^{i\alpha}\rmd\eta^{\dot\alpha}_i-\rmd\theta^{i\alpha}\eta^{\dot\alpha}_i)\hat\mu_\alpha\hat\lambda_{\dot\alpha}}{|\mu|^2|\lambda|^2}~,~~~
                    \hat e^\rmL\ \coloneqq\ \frac{\rmd\hat\lambda_{\dot\alpha}\hat\lambda^{\dot\alpha}}{|\lambda|^4}~,~~~
                    \hat e^\rmR\ \coloneqq\ \frac{\rmd\hat\mu_\alpha\hat\mu^\alpha}{|\mu|^4}~,
                    \\
                    \hat e_i\ \coloneqq\ -\frac{\rmd\eta^{\dot\alpha}_i\hat\lambda_{\dot\alpha}}{|\lambda|^2}~,~~~
                    \hat e^i\ \coloneqq\ -\frac{\rmd\theta^{i\alpha}\hat\mu_\alpha}{|\mu|^2}~,
                \end{gathered}
            \end{equation}
        \end{subequations}
        where we have slightly abused notation and again denoted the dual of $\hat E_\rmF$ by $\hat e^\rmF$. Note that the only non-vanishing commutator amongst the vector fields~\eqref{eq:tangentBundleDistributionsCRSuperAmbitwistorSpace} is
        \begin{equation}\label{eq:superspaceTorsion}
            [\hat E_i,\hat E^j]\ =\ 2\delta_i{}^j\hat E_\rmF~, 
        \end{equation}
        where $\delta_i{}^j$ is the Kronecker symbol. Next, set $\Omega^{0,\bullet}_{\rm CR}(F)\coloneqq\bigwedge^\bullet T^{0,1\,*}_{\rm CR}F$ for the distribution $T^{0,1}_{\rm CR}F$ defined in~\eqref{eq:tangentBundleDistributionsCRSuperAmbitwistorSpace}.\footnote{with independence on the coordinates $\bar\eta^{\dot\alpha}_i$ and $\bar\theta^{i\alpha}$ implied, as stated above} This is augmented to a differential graded algebra via the differential
        \begin{equation}
            \bar\partial_{\rm CR}\ \coloneqq\ \hat e^\rmF\hat E_\rmF+\hat e^\rmL\hat E_\rmL+\hat e^\rmR\hat E_\rmR+\hat e_i\hat E^i+\hat e^i\hat E_i~.
        \end{equation}
        We also have
        \begin{subequations}
            \begin{equation}
                \begin{gathered}
                    T^{1,0}_{\rm CR}F\ \coloneqq\ {\rm span}\{E_\rmF,E_\rmL,E_\rmR,E_\rmW,E_{\hat\rmW},E^i,E_i\}~,
                    \\
                    E_\rmF\ \coloneqq\ \frac{\hat\mu^\alpha\hat\lambda^{\dot\alpha}}{|\mu|^2|\lambda|^2}\parder{x^{\alpha\dot\alpha}}~,~~~
                    E_\rmL\ \coloneqq\ -\frac{\hat\lambda_{\dot\alpha}}{|\lambda|^2}\parder{\lambda_{\dot\alpha}}~,~~~
                    E_\rmR\ \coloneqq\ -\frac{\hat\mu_\alpha}{|\mu|^2}\parder{\mu_\alpha}~,
                    \\
                    E_\rmW\ \coloneqq\ \frac{\mu^\alpha\hat\lambda^{\dot\alpha}}{|\lambda|^2}\parder{x^{\alpha\dot\alpha}}~,~~~
                    E_{\hat\rmW}\ \coloneqq\ -\frac{\hat\mu^\alpha\lambda^{\dot\alpha}}{|\mu|^2}\parder{x^{\alpha\dot\alpha}}~,
                    \\
                    E^i\ \coloneqq\ \frac{\hat\lambda^{\dot\alpha}}{|\lambda|^2}D^i_{\dot\alpha},~~~
                    E_i\ \coloneqq\ \frac{\hat\mu^\alpha}{|\mu|^2}D_{i\alpha}
                \end{gathered}
            \end{equation}
            with $D^i_{\dot\alpha}$ and $D_{i\alpha}$ as given in~\eqref{eq:tangentBundleDistributionsCRSuperAmbitwistorSpace}. Dually,
            \begin{equation}
                \begin{gathered}
                    \Omega^{1,0}_{\rm CR}(F)\ \coloneqq\ {\rm span}\{e^\rmF,e^\rmL,e^\rmR,e^\rmW,e^{\hat\rmW},e_i,e^i\}~,
                    \\
                    e^\rmF\ \coloneqq\ (\rmd x^{\alpha\dot\alpha}+\theta^{i\alpha}\rmd\eta^{\dot\alpha}_i-\rmd\theta^{i\alpha}\eta^{\dot\alpha}_i)\mu_\alpha\lambda_{\dot\alpha}~,~~~
                    e^\rmL\ \coloneqq\ \rmd\lambda_{\dot\alpha}\lambda^{\dot\alpha}~,~~~
                    e^\rmR\ \coloneqq\ \rmd\mu_\alpha\mu^\alpha~,
                    \\
                    e^\rmW\ \coloneqq\ -\frac{(\rmd x^{\alpha\dot\alpha}+\theta^{i\alpha}\rmd\eta^{\dot\alpha}_i-\rmd\theta^{i\alpha}\eta^{\dot\alpha}_i)\hat\mu_\alpha\lambda_{\dot\alpha}}{|\mu|^2}~,~~~
                    e^{\hat\rmW}\ \coloneqq\ \frac{(\rmd x^{\alpha\dot\alpha}+\theta^{i\alpha}\rmd\eta^{\dot\alpha}_i-\rmd\theta^{i\alpha}\eta^{\dot\alpha}_i)\mu_\alpha\hat\lambda_{\dot\alpha}}{|\lambda|^2}~,
                    \\
                    e_i\ \coloneqq\ \rmd\eta^{\dot\alpha}_i\lambda_{\dot\alpha}~,~~~
                    e^i\ \coloneqq\ \rmd\theta^{i\alpha}\mu_\alpha~,
                \end{gathered}
            \end{equation}
            where, as before, we have slightly abused notation in labelling our one-forms. Hence,
            \begin{equation}
                \partial_{\rm CR}\ \coloneqq\ e^\rmF E_\rmF+e^\rmL E_\rmL+e^\rmR E_\rmR+e^\rmW E_\rmW+e^{\hat\rmW}E_{\hat\rmW}+e_iE^i+e^iE_i
            \end{equation}
        \end{subequations}
        with $\partial_{\rm CR}^2=0$, which implies that the distribution $T^{1,0}_{\rm CR}F$ is integrable as well. Furthermore,
        \begin{equation}
            \begin{aligned}
                \rmd\ &=\ \bar\partial_{\rm CR}+\partial_{\rm CR}
                \\ 
                &=\ \underbrace{\rmd x^{\alpha\dot\alpha}\parder{x^{\alpha\dot\alpha}}+\rmd\eta^{\dot\alpha}_i\parder{\eta^{\dot\alpha}_i}+\rmd\theta^{i\alpha}\parder{\theta^{i\alpha}}}_{=\,\rmd_{\IR^{4|12}_{\rm cpl}}}+\underbrace{\hat e^\rmL\hat E_\rmL+e^\rmL E_\rmL+\hat e^\rmR\hat E_\rmR+e^\rmR E_\rmR}_{=\,\rmd_{\IC P^1\times\IC P^1}}
            \end{aligned}
        \end{equation}
        for the exterior derivative on $F$.\footnote{As mentioned before, we restrict to functions and differential forms on $F$ that depend on the fermionic coordinates only via $\eta^{\dot\alpha}_i$ and $\theta^{i\alpha}$ and not their complex conjugates. Hence, the fermionic part of $\rmd$ only contains $\rmd\eta^{\dot\alpha}_i\parder{\eta^{\dot\alpha}_i}+\rmd\theta^{i\alpha}\parder{\theta^{i\alpha}}$ and not its complex conjugate.}
        
        \paragraph{Constraint system of $\caN=3$ supersymmetric Yang--Mills theory.}
        Let now $\frg$ be a Lie algebra with Lie bracket $[-,-]$ and consider $A\in\Omega^{0,1}_{\rm CR}(F)\otimes\frg$ for a topologically trivial complex vector bundle over $F$. The equation of motion for \uline{CR holomorphic Chern--Simons theory}\footnote{In~\cite{Popov:2005uv}, such a theory is referred to as partially holomorphic Chern--Simons theory.} reads 
        \begin{equation}\label{eq:CRCSEoM}
            \bar\partial_{\rm CR}A+\tfrac12[A,A]\ =\ 0
        \end{equation}
        with the wedge product in the bracket $[-,-]$ understood. Under the standard (twistor) assumption that there is a gauge of $A$ in which\footnote{The symbol `$\intprod$' denotes the interior product.} 
        \begin{equation}\label{eq:ambitwistorGauge}
            \hat E_\rmL\intprod A\ =\ 0\ =\ \hat E_\rmR\intprod A~,
        \end{equation}
        we deduce from~\eqref{eq:CRCSEoM} that the remaining components of $A$ are holomorphic in $\lambda_{\dot\alpha}$ and $\mu_\alpha$. Following~\cite{Manin:1988ds}, we call topologically trivial complex vector bundles over $F$ with a connection that allows for such a gauge \uline{$\IR^4$-trivial}.\footnote{Put differently, all such bundles are given by pull-backs along $F\rightarrow\IR^4$.} Consequently,~\eqref{eq:CRCSEoM} is equivalent to the equations\footnote{Parentheses denote normalised total symmetrisation of the enclosed indices.}
        \begin{subequations}
            \begin{equation}
                \big[\nabla^i_{(\dot\alpha},\nabla^j_{\dot\beta)}\big]\ =\ 0~,~~~
                \big[\nabla_{i(\alpha},\nabla_{j\beta)}\big]\ =\ 0~,~~~
                \big[\nabla_{i\alpha},\nabla^j_{\dot\alpha}\big]\ =\ 2\delta_i{}^j\nabla_{\alpha\dot\alpha}
            \end{equation}
            with\footnote{$\partial_{\alpha\dot\alpha}\coloneqq\parder{x^{\alpha\dot\alpha}}$.}
            \begin{equation}
                \nabla^i_{\dot\alpha}\ \coloneqq\ D^i_{\dot\alpha}+[A^i_{\dot\alpha},-]~,~~~
                \nabla_{i\alpha}\ \coloneqq D_{i\alpha}+[A_{i\alpha},-]~,~~~
                \nabla_{\alpha\dot\alpha}\ \coloneqq\ \partial_{\alpha\dot\alpha}+[A_{\alpha\dot\alpha},-]
            \end{equation}
        \end{subequations}
        which, in turn, constitute the superspace constraint system of the equations of motion of $\caN=3$ supersymmetric Yang--Mills theory on $\IR^4$ and with gauge algebra $\frg$. See~\cite{Witten:1978xx,Isenberg:1978kk,Khenkin:1980ff,Harnad:1984vk,Harnad:1985bc,Buchdahl:1985aa,Eastwood:1987aa,Manin:1988ds} (and also~\cite{Popov:2004rb} for a review) for full details on this construction. 
        
        It is also straightforward to check that infinitesimal gauge transformations of $\caN=3$ supersymmetric Yang--Mills theory are captured by $c\in\Omega^{0,0}_{\rm CR}(F)\otimes\frg$ and act according to
        \begin{equation}
            A\ \mapsto\ A+\delta A
            \ewith
            \delta A\ \coloneqq\ \bar\partial_{\rm CR}c+[A,c]~.
        \end{equation} 
        
        We thus see that the at the level of equations of motion, $\caN=3$ supersymmetric Yang--Mills theory can be identified with the differential graded Lie algebra\footnote{As before, the wedge product in the bracket $[-,-]$ is understood.} 
        \begin{equation}\label{eq:SYMDGLA}
            \frL_{\rm CR}\ \coloneqq\ \big(\Omega^{0,\bullet}_{\rm CR}(F)\otimes\frg,\bar\partial_{\rm CR},[-,-]\big)\,,
        \end{equation}
        where elements of degrees $0$ and $1$ correspond to the gauge parameters and gauge potentials. Elements of degree $2$ and $3$ correspond to the anti-fields of the gauge potentials and the anti-fields of the ghosts. We have thus obtained the Batalin--Vilkovisky formulation of the CR holomorphic Chern--Simons theory, following the familiar construction for ordinary Chern--Simons theory, at the level of equations of motion.
        
        Recall that any equation of motion can be brought into the form of a homotopy Maurer--Cartan equation associated with an $L_\infty$-algebra, the natural generalisation of differential graded Lie algebras with products of higher arities. Moreover, \uline{semi-classical equivalence}\footnote{i.e.~the theories have isomorphic tree-level S-matrices} corresponds to a \uline{quasi-isomorphism} between $L_\infty$-algebras. This perspective will prove to be very useful to our discussion; a detailed review can be found e.g.~in~\cite{Jurco:2018sby} and a discussion of field theory equivalences is found in~\cite[Section 3.4]{Borsten:2021gyl}. 
        
        By construction, any $L_\infty$-algebra comes with an underlying cochain complex, and any quasi-isomorphism of $L_\infty$-algebras descends to a quasi-isomorphism of the underlying cochain complexes, that is, a cochain map that descend to an isomorphism on the cohomologies. Therefore, we are interested in the following proposition, which follows from our above considerations.

        \begin{proposition}\label{prop:quasi-iso-first-cohomologies}
            Under the assumption of $\IR^4$-triviality, the cohomology groups $H^0(\frL_{\rm CR})$ and $H^1(\frL_{\rm CR})$ are isomorphic to the trivial (i.e.~constant) linearised gauge transformations and the quotient space of solutions to the linearised $\caN=3$ supersymmetric Yang--Mills equations modulo linearised gauge transformations, respectively.
        \end{proposition}

        \noindent 
        We will discuss the remaining cohomology groups in \cref{ssec:equivalence}.
        
        \subsection{Quasi-isomorphic differential graded Lie algebras}
        
        The form of the CR structure~\eqref{eq:tangentBundleDistributionsCRSuperAmbitwistorSpace} is not quite suitable for writing down an action. The main problem is as follows. The gauge potential $A$ includes the fermionic differential form components $\hat E^i\intprod A$ and $\hat E_i\intprod A$. However, the Berezin integration over fermionic coordinates requires integral forms, rather than differential forms. Therefore, to make sense of the Berezin integration, one might be tempted to work in a gauge in which the fermionic directions $\hat E^i\intprod A$ and $\hat E_i\intprod A$ vanish; see~\cite{Witten:2003nn} for the chiral setting. However, because of the superspace torsion that is encoded in the non-trivial commutator~\eqref{eq:superspaceTorsion}, this would imply the vanishing of the bosonic component $\hat E_\rmF\intprod A$. In turn, we would be left with too few bosonic directions (in fact, only two) to formulate a Chern--Simons-type action. This is in contradistinction with supersymmetric self-dual Yang--Mills theory where one can avoid this issue by working in chiral superspace. To resolve this problem, we shall replace the differential graded Lie algebra~\eqref{eq:SYMDGLA} with a quasi-isomorphic one, 
        \begin{equation}\label{eq:SYMDGLATwisted}
            \frL_{\rm CR,\,tw}\ \coloneqq\ \big(\Omega^{0,\bullet}_{\rm CR,\,tw}(F)\otimes\frg,\bar\partial_{\rm CR,\,tw},[-,-]\big)\,,
        \end{equation}
        which we develop in the following.
        
        \paragraph{Twisted CR structure.}
        For brevity, we use the \uline{CR holomorphic} and \uline{CR anti-holomorphic fermionic coordinates}\footnote{Whilst $\eta_i$ and $\tilde\eta_i$ and likewise $\theta^i$ and $\tilde\theta^i$ are not complex conjugates of each other and are not related by the quaternionic structure~\eqref{eq:quaternionicStructure}, we have $\bar\partial_{\rm CR}\eta_i=\bar\partial_{\rm CR}\theta^i=0$ and $\partial_{\rm CR}\tilde\eta_i=\partial_{\rm CR}\tilde\theta^i=0$.} 
        \begin{subequations}\label{eq:CRFermionicCoordinatesTransformations}
            \begin{equation}\label{eq:CRFermionicCoordinates}
                \eta_i\ \coloneqq\ \eta^{\dot\alpha}_i\lambda_{\dot\alpha}~,~~~
                \theta^i\ \coloneqq\ \theta^{i\alpha}\mu_\alpha~,~~~
                \tilde\eta_i\ \coloneqq\ -\frac{\eta^{\dot\alpha}_i\hat\lambda_{\dot\alpha}}{|\lambda|^2}~,
                \eand
                \tilde\theta^i\ \coloneqq\ -\frac{\theta^{i\alpha}\hat\mu_\alpha}{|\mu|^2}
            \end{equation}
            with the inverse relations
            \begin{equation}
                \eta^{\dot\alpha}_i\ =\ \tilde\eta_i\lambda^{\dot\alpha}+\frac{\eta^i\hat\lambda^{\dot\alpha}}{|\lambda|^2}
                \eand
                \theta^{i\alpha}\ =\ \tilde\theta^i\mu^\alpha+\frac{\theta^i\hat\mu^\alpha}{|\mu|^2}
            \end{equation}
        \end{subequations}
        and consider the change of basis 
        \begin{subequations}\label{eq:CRStructureSuperAmbitwistorSpaceRedefined}
            \begin{equation}\label{eq:tangentBundleDistributionsCRSuperAmbitwistorSpaceRedefined}
                \begin{gathered}
                    T^{0,1}_{\rm CR}F\ =\ {\rm span}\{\hat E'_\rmF,\hat E'_\rmL,\hat E'_\rmR,\hat E'^i,\hat E'_i\}~,
                    \\
                    \hat E'_\rmF\ \coloneqq\ \hat E_\rmF~,~~~
                    \hat E'_\rmL\ \coloneqq\ \hat E_\rmL+\tilde\theta^i\eta_i\hat E_\rmF~,~~~
                    \hat E'_\rmR\ \coloneqq\ \hat E_\rmR-\theta^i\tilde\eta_i\hat E_\rmF~,
                    \\
                    \hat E'^i\ \coloneqq\ \hat E^i-\tilde\theta^i\hat E_\rmF~,~~~
                    \hat E'_i\ \coloneqq\ \hat E_i-\tilde\eta_i\hat E_\rmF
                \end{gathered}
            \end{equation}
            of the CR structure. We stress that these vector fields respect the holomorphic dependence of $\Omega^{0,\bullet}_{\rm CR}(F)$ on $\eta^{\dot\alpha}_i$ and $\theta^{i\alpha}$. Dually, we have
            \begin{equation}
                \begin{gathered}
                    \Omega^{0,1}_{\rm CR}(F)\ \coloneqq\ {\rm span}\{\hat e'^\rmF,\hat e'^\rmL,\hat e'^\rmR,\hat e'_i,\hat e'^i\}~,
                    \\
                    \hat e'^\rmF\ \coloneqq\ \hat e^\rmF-\tilde\theta^i\eta_i\hat e^\rmL+\theta^i\tilde\eta_i\hat e^\rmR-\tilde\theta_i\hat e_i-\tilde\eta_i\hat e^i~,
                    \\
                    \hat e'^\rmL\ \coloneqq\ \hat e^\rmL~,~~~
                    \hat e'^\rmR\ \coloneqq\ \hat e^\rmR~,~~~
                    \hat e'_i\ \coloneqq\ \hat e_i~,~~~
                    \hat e'^i\ \coloneqq\ \hat e^i~.
                \end{gathered}
            \end{equation}
        \end{subequations}
        It is then not difficult to see that the only non-vanishing commutator amongst the vector fields~\eqref{eq:tangentBundleDistributionsCRSuperAmbitwistorSpaceRedefined} is
        \begin{equation}\label{eq:superspaceTorsionRedefined}
            [\hat E'_\rmL,\hat E'_\rmR]\ =\ 2\theta^i\eta_i\hat E'_\rmF~.
        \end{equation}
        Hence, the superspace torsion initially encoded in the fermionic commutator~\eqref{eq:superspaceTorsion} has been shifted to a bosonic commutator for the auxiliary spinorial coordinates. 
        
        Next, to remove the dependence on the CR anti-holomorphic fermionic coordinates $\tilde\eta_i$ and $\tilde\theta^i$, let
        \begin{subequations}
            \begin{equation}\label{eq:twist}
                g\ \coloneqq\ \rme^{\tilde\theta^i\eta_iE_\rmW+\theta^i\tilde\eta_iE_{\hat\rmW}}
            \end{equation}
            with $E_\rmW$ and $E_{\hat\rmW}$ as given in~\eqref{eq:tangentBundleDistributionsCRAmbitwistorSpace}, and consider the \uline{twisted} almost CR structure
            \begin{equation}\label{eq:tangentBundleDistributionsCRSuperAmbitwistorSpaceTwisted}
                \begin{gathered}
                    T^{0,1}_{\rm CR,\,tw}F\ \coloneqq\ {\rm span}\{\hat V_\rmF,\hat V_\rmL,\hat V_\rmR,\hat V^i,\hat V_i\}~,
                    \\
                    \hat V_\rmF\ \coloneqq\ g\hat E'_\rmF g^{-1}\ =\ \hat E_\rmF~,
                    \\
                    \hat V_\rmL\ \coloneqq\ g\hat E'_\rmL g^{-1}\ =\ \hat E_\rmL+\theta^i\eta_iE_{\hat\rmW}~,~~~
                    \hat V_\rmR\ \coloneqq\ g\hat E'_\rmR g^{-1}\ =\ \hat E_\rmR+\theta^i\eta_iE_\rmW~,
                    \\
                    \hat V^i\ \coloneqq\ g\hat E'^ig^{-1}\ =\ \hat E^i-\tilde\theta^i\hat E_\rmF+\theta^iE_{\hat\rmW}\ =\ \lambda^{\dot\alpha}\parder{\eta^{\dot\alpha}_i}~,
                    \\
                    \hat V_i\ \coloneqq\ g\hat E'_ig^{-1}\ =\ \hat E_i-\tilde\eta_i \hat E_\rmF-\eta_iE_\rmW\ =\ \mu^\alpha\parder{\theta^{i\alpha}}~.
                \end{gathered}
            \end{equation}
            One may quickly check that this almost CR structure is again integrable. It is important to stress that after the twisting, the fermionic vector fields $\hat V^i$ and $\hat V_i$ have become fermionic derivatives with respect to the CR anti-holomorphic coordinates $\tilde\eta_i$ and $\tilde\theta^i$, and, in addition, the bosonic vector fields depend only on the CR holomorphic coordinates $\eta_i$ and $\theta^i$.\footnote{Under the change of coordinates $(\lambda_{\dot\alpha},\eta^{\dot\alpha}_i)\mapsto(\pi_{\dot\alpha},\eta_i,\tilde\eta_i)$ and $(\mu_\alpha,\theta^{i\alpha})\mapsto(\rho_\alpha,\theta^i,\tilde\theta^i)$ with $\pi_{\dot\alpha}\coloneqq\lambda_{\dot\alpha}$, $\rho_\alpha\coloneqq\mu_\alpha$, and~\eqref{eq:CRFermionicCoordinates}, it is not difficult to see that $\hat V^i=\parder{\tilde\eta_i}$ and $\hat V_\rmL=|\pi|^2\pi_{\dot\alpha}\parder{\hat\pi_{\dot\alpha}}-\eta_i\parder{\tilde\eta_i}+\theta^i\eta_iE_{\hat\rmW}$ as well as $\hat V_i=\parder{\tilde\theta^i}$ and $\hat V_\rmR=|\rho|^2\rho_\alpha\parder{\hat\rho_\alpha}-\theta^i\parder{\tilde\theta^i}+\theta^i\eta_i E_\rmW$.\label{fn:changeOfFermionicCoordinates}} Dually, we have
            \begin{equation}\label{eq:cotangentBundleDistributionCRSuperAmbitwistorSpaceTwisted}
                \begin{gathered}
                    \Omega^{0,1}_{\rm CR,\,tw}(F)\ \coloneqq\ {\rm span}\{\hat v^\rmF,\hat v^\rmL,\hat v^\rmR,\hat v_i,\hat v^i\}~,
                    \\
                    \hat v^\rmF\ \coloneqq\ \hat e^\rmF-\tilde\theta^i\hat e_i-\tilde\eta_i\hat e^i\ =\ \frac{\rmd x^{\alpha\dot\alpha}\hat\mu_\alpha\hat\lambda_{\dot\alpha}}{|\mu|^2|\lambda|^2}~,
                    \\
                    \hat v^\rmL\ \coloneqq\ \hat e^\rmL\ =\ \frac{\rmd\hat\lambda_{\dot\alpha}\hat\lambda^{\dot\alpha}}{|\lambda|^4}~,~~~
                    \hat v^\rmR\ \coloneqq\ \hat e^\rmR\ =\ \frac{\rmd\hat\mu_\alpha\hat\mu^\alpha}{|\mu|^4}~,
                    \\
                    \hat v_i\ \coloneqq\ \hat e_i\ =\ -\frac{\rmd\eta^{\dot\alpha}_i\hat\lambda_{\dot\alpha}}{|\lambda|^2}~,~~~
                    \hat v^i\ \coloneqq\ \hat e^i\ =\ -\frac{\rmd\theta^{i\alpha}\hat\mu_\alpha}{|\mu|^2}~.
                \end{gathered}
            \end{equation}
        \end{subequations}
        Upon setting $\Omega^{0,\bullet}_{\rm CR,\,tw}(F)\coloneqq\bigwedge^\bullet T^{0,1\,*}_{\rm CR,\,tw}F$, we obtain the differential graded algebra $(\Omega^{0,\bullet}_{\rm CR,\,tw}(F),\bar\partial_{\rm CR,\,tw})$ with
        \begin{equation}\label{eq:CRantiholomorphicDifferentialTwisted}
            \bar\partial_{\rm CR,\,tw}\ \coloneqq\ \hat v^\rmF\hat V_\rmF+\hat v^\rmL\hat V_\rmL+\hat v^\rmR\hat V_\rmR+\hat v_i\hat V^i+\hat v^i\hat V_i~.
        \end{equation}
        We also have
        \begin{subequations}
            \begin{equation}\label{eq:tangentBundleDistributionsCRSuperAmbitwistorSpaceTwistedHolomorphic}
                \begin{gathered}
                    T^{1,0}_{\rm CR,\,tw}F\ \coloneqq\ {\rm span}\{V_\rmF,V_\rmL,V_\rmR,V_\rmW,V_{\hat\rmW},V^i,V_i\}~,
                    \\
                    V_\rmF\ \coloneqq\ E_\rmF~,~~~
                    V_\rmL\ \coloneqq\ E_\rmL-\tilde\theta^i\tilde\eta_iE_\rmW~,~~~
                    V_\rmR\ \coloneqq\ E_\rmR-\tilde\theta^i\tilde\eta_iE_{\hat\rmW}~,
                    \\
                    V_\rmW\ \coloneqq\ E_\rmW~,~~~
                    V_{\hat\rmW}\ \coloneqq\ E_{\hat\rmW}~,
                    \\
                    V^i\ \coloneqq\ E^i-\theta^iE_\rmF-\tilde\theta^iE_\rmW\ =\ \hat\lambda^{\dot\alpha}\parder{\eta^{\dot\alpha}_i}~,
                    \\
                    V_i\ \coloneqq\ E_i-\eta_iE_\rmF+\tilde\eta_iE_{\hat\rmW}\ =\ \hat\mu^\alpha\parder{\theta^{i\alpha}}~,
                \end{gathered}
            \end{equation}
            and, dually,
            \begin{equation}\label{eq:holomorphicCotangentBundleDistributionCRSuperAmbitwistorSpaceTwisted}
                \begin{gathered}
                    \Omega^{1,0}_{\rm CR,\,tw}(F)\ \coloneqq\ {\rm span}\{v^\rmF,v^\rmL,v^\rmR,v^\rmW,v^{\hat\rmW},v_i,v^i\}~,
                    \\
                    v^\rmF\ \coloneqq\ e^\rmF-\theta^ie_i-\eta_ie^i\ =\ \rmd x^{\alpha\dot\alpha}\mu_\alpha\lambda_{\dot\alpha}~,
                    \\
                    v^\rmL\ \coloneqq\ e^\rmL\ =\ \rmd\lambda_{\dot\alpha}\lambda^{\dot\alpha}~,~~~
                    v^\rmR\ \coloneqq\ e^\rmR\ =\ \rmd\mu_\alpha\mu^\alpha~,
                    \\
                    v^\rmW\ \coloneqq\ e^\rmW-\eta_i\hat e^i-\tilde\theta^ie_i-\theta^i\eta_i\hat e^\rmR+\tilde\theta^i\tilde\eta_ie^\rmL\ =\ -\frac{\rmd x^{\alpha\dot\alpha}\hat\mu_\alpha\lambda_{\dot\alpha}}{|\mu|^2}-\theta^i\eta_i\frac{\rmd\hat\mu_\alpha\hat\mu^\alpha}{|\mu|^4}+\tilde\theta^i\tilde\eta_i\rmd\lambda_{\dot\alpha}\lambda^{\dot\alpha}~,
                    \\
                    v^{\hat\rmW}\ \coloneqq\ e^{\hat\rmW}+\theta^i\hat e_i+\tilde\eta_ie^i-\theta^i\eta_i\hat e^\rmL+\tilde\theta^i\tilde\eta_ie^\rmR\ =\ \frac{\rmd x^{\alpha\dot\alpha}\mu_\alpha\hat\lambda_{\dot\alpha}}{|\lambda|^2}-\theta^i\eta_i\frac{\rmd\hat\lambda_{\dot\alpha}\hat\lambda^{\dot\alpha}}{|\lambda|^4}+\tilde\theta^i\tilde\eta_i\rmd\mu_\alpha\mu^\alpha~, 
                    \\
                    v_i\ \coloneqq\ e_i\ =\ \rmd\eta^{\dot\alpha}_i\lambda_{\dot\alpha}~,~~~
                    v^i\ \coloneqq\ e^i\ =\ \rmd\theta^{i\alpha}\mu_\alpha
                \end{gathered}
            \end{equation}
            so that
            \begin{equation}
                \partial_{\rm CR,\,tw}\ \coloneqq\ v^\rmF V_\rmF+v^\rmL V_\rmL+v^\rmR V_\rmR+v^\rmW V_\rmW+v^{\hat\rmW}V_{\hat\rmW}+v_i V^i+v^iV_i~,
            \end{equation}
        \end{subequations}
        as well as $\rmd=\bar\partial_{\rm CR,\,tw}+\partial_{\rm CR,\,tw}$ for the exterior derivative on $F$; note that also $\partial_{\rm CR,\,tw}^2=0$, which tells us that the distribution $T^{1,0}_{\rm CR,\,tw}F$ is integrable as well.
        
        \paragraph{Quasi-isomorphic CR complex.}
        Importantly, the only non-vanishing commutator amongst the vector fields~\eqref{eq:tangentBundleDistributionsCRSuperAmbitwistorSpaceTwisted} is 
        \begin{equation}\label{eq:superspaceTorsionTwisted}
            [\hat V_\rmL,\hat V_\rmR]\ =\ 2\theta^i\eta_i\hat V_\rmF
        \end{equation}
        which is the same as~\eqref{eq:superspaceTorsionRedefined} and so, the twist in~\eqref{eq:tangentBundleDistributionsCRSuperAmbitwistorSpaceTwisted} mediated by $g$ is a twist by an (outer) automorphism of the CR structure. This implies the following result.
        
        \begin{proposition}\label{prop:quasiIsoTwisted}
            The cochain complexes $(\Omega^{0,\bullet}_{\rm CR,\,tw}(F),\bar\partial_{\rm CR,\,tw})$ and $(\Omega^{0,\bullet}_{\rm CR}(F),\bar\partial_{\rm CR})$ are (quasi-)isomorphic.
        \end{proposition}
        
        \begin{proof}
            To describe $(\Omega^{0,\bullet}_{\rm CR}(F),\bar\partial_{\rm CR})$, we work in the basis~\eqref{eq:CRStructureSuperAmbitwistorSpaceRedefined}. In particular, let us denote the basis~\eqref{eq:tangentBundleDistributionsCRSuperAmbitwistorSpaceRedefined} collectively as $\hat E'_A$ with $A,B,\ldots$ multi-indices. The relation between this basis and the basis~\eqref{eq:tangentBundleDistributionsCRSuperAmbitwistorSpaceTwisted} is then $\hat V_A=g\hat E'_Ag^{-1}$, and by virtue of~\eqref{eq:superspaceTorsionRedefined} and~\eqref{eq:superspaceTorsionTwisted}, we also obtain
            \begin{equation}
                [\hat V_A,\hat V_B]\ =\ C_{AB}{}^C\hat V_C
                \eand
                [\hat E'_A,\hat E'_B]\ =\ C_{AB}{}^C\hat E'_C
            \end{equation}
            with $gC_{AB}{}^C=C_{AB}{}^Cg$ $\Leftrightarrow$ $g^{-1}C_{AB}{}^C=C_{AB}{}^Cg^{-1}$. For $f$ a function, this implies
            \begin{equation}
                \hat V_Af\ =\ 0
                \quad\Leftrightarrow\quad
                \hat E'_A(g^{-1}f)\ =\ 0
            \end{equation}
            and so, the zeroth cohomology groups of $(\Omega^{0,\bullet}_{\rm CR,\,tw}(F),\bar\partial_{\rm CR,\,tw})$ and $(\Omega^{0,\bullet}_{\rm CR}(F),\bar\partial_{\rm CR})$ are isomorphic by means of $f\mapsto g^{-1}f$. Likewise, for $\alpha\in\Omega^{0,1}_{\rm CR,\,tw}(F)$, we obtain
            \begin{subequations}
                \begin{equation}
                    \hat V_A\alpha_B-(-1)^{|A||B|}\hat V_B\alpha_A-C_{AB}{}^C\alpha_C\ =\ 0
                \end{equation}
                for the $\bar\partial_{\rm CR,\,tw}$-closure condition and since $g^{-1}C_{AB}{}^C=C_{AB}{}^Cg^{-1}$, this is equivalent to
                \begin{equation}
                    \hat E'_A(g^{-1}\alpha_B)-(-1)^{|A||B|}\hat E'_B(g^{-1}\alpha_A)-C_{AB}{}^C(g^{-1}\alpha_C)\ =\ 0~.
                \end{equation}
                Since also 
                \begin{equation}
                    \alpha_A\ \mapsto\ \alpha_A+\hat V_Af
                    \quad\Leftrightarrow\quad
                    g^{-1}\alpha_A\ \mapsto\ g^{-1}\alpha_A+\hat E'_A(g^{-1}f)~,
                \end{equation}
            \end{subequations}
            we thus conclude that the first cohomology groups of $(\Omega^{0,\bullet}_{\rm CR,\,tw}(F),\bar\partial_{\rm CR,\,tw})$ and $(\Omega^{0,\bullet}_{\rm CR}(F),\bar\partial_{\rm CR})$ are isomorphic by means of $(\alpha_A,f)\mapsto(g^{-1}\alpha_A,g^{-1}f)$. The general case now follows straightforwardly.
        \end{proof}
        
        \paragraph{Quasi-isomorphic differential graded Lie algebras.} 
        \cref{prop:quasiIsoTwisted} shows that the asymptotically free fields\footnote{i.e.~physically, the labels in would-be scattering amplitudes and mathematically, the cohomologies of the CR complexes.} of the CR holomorphic Chern--Simons equations for the differential graded Lie algebras $\frL_{\rm CR}$ and $\frL_{\rm CR,\,tw}$ defined in~\eqref{eq:SYMDGLA} and~\eqref{eq:SYMDGLATwisted} are isomorphic. For a full equivalence, it remains to show the following.
        
        \begin{proposition}\label{prop:fullQuasiIso}
            The CR holomorphic Chern--Simons equations of motion\footnote{physically: in their Batalin--Vilkovisky formulation including ghosts and anti-fields; see the discussion following~\eqref{eq:SYMDGLA}.}  defined by the differential graded Lie algebras $\frL_{\rm CR}$ and $\frL_{\rm CR,\,tw}$ from~\eqref{eq:SYMDGLA} and~\eqref{eq:SYMDGLATwisted} are equivalent, that is, $\frL_{\rm CR}$ and $\frL_{\rm CR,\,tw}$ are (quasi-)isomorphic as differential graded Lie algebras.
        \end{proposition}
        
        \begin{proof}
            The quasi-isomorphism in \cref{prop:quasiIsoTwisted} is mediated by $g$ as defined in~\eqref{eq:twist}, and we now need to show that the wedge products are mapped consistently into each other under this quasi-isomorphism. This, however, is a direct consequence of the identity
            \begin{equation}
                g[g^{-1}\alpha_{A_1A_2\cdots},g^{-1}\beta_{B_1B_2\cdots}]\ =\ [\alpha_{A_1A_2\cdots},\beta_{B_1B_2\cdots}]
            \end{equation}
            for any $\alpha,\beta\in\Omega^{0,\bullet}_{\rm CR,\,tw}(F)\otimes\frg$ where we have used the notation of the proof of \cref{prop:quasiIsoTwisted}. Note that this identity follows straightforwardly upon setting $c(t)\coloneqq[\rme^{-tX}\alpha_{A_1A_2\cdots},\rme^{-tX}\beta_{B_1B_2\cdots}]$ for all $t\in\IR$ and with $X\coloneqq\tilde\theta^i\eta_iE_\rmW+\theta^i\tilde\eta_iE_{\hat\rmW}$ the vector field in the definition~\eqref{eq:twist} of $g$. Then, we obtain the differential equation $\dot c(t)=-Xc(t)$ which has the general solution $c(t)=\rme^{-tX}c(0)$. The identity now follows for $t=1$.
        \end{proof}
        
        \noindent
        Recall the notion of $\IR^4$-triviality from around~\eqref{eq:ambitwistorGauge}.
        
        \begin{corollary}
            For an $\IR^4$-trivial complex vector bundle over the augmented CR ambitwistor space $F$, the \uline{twisted CR holomorphic Chern--Simons equation}
            \begin{equation}\label{eq:CRCSEoMTwisted}
                \bar\partial_{\rm CR,\,tw}A+\tfrac12[A,A]\ =\ 0
            \end{equation}
            for $A\in\Omega^{0,1}_{\rm CR,\,tw}(F)\otimes\frg$ is equivalent to the equations of motion of $\caN=3$ supersymmetric Yang--Mills theory on Euclidean space $\IR^4$.
        \end{corollary}
        
        \begin{proof}
            This is a direct consequence of \cref{prop:fullQuasiIso} and our discussion around~\eqref{eq:CRCSEoM}.
        \end{proof}
        
        \paragraph{Witten gauge.}
        The graded algebra $\Omega^{0,\bullet}_{\rm CR,\,tw}(F)$ is generated by the differential forms~\eqref{eq:cotangentBundleDistributionCRSuperAmbitwistorSpaceTwisted}, and the coefficient functions generically depend on all the fermionic coordinates $\eta^{\dot\alpha}_i$ and $\theta^{i\alpha}$ or, equivalently, on $(\eta_i,\tilde\eta_i)$ and $(\theta^i,\tilde\theta^i)$ via~\eqref{eq:CRFermionicCoordinates}. We now let $\Omega^{0,\bullet}_{\rm CR,\,tw,\,red}(F)\subseteq\Omega^{0,\bullet}_{\rm CR,\,tw}(F)$ be the graded subalgebra of $\Omega^{0,\bullet}_{\rm CR,\,tw}(F)$ that is generated by only the bosonic differential forms from~\eqref{eq:cotangentBundleDistributionCRSuperAmbitwistorSpaceTwisted} and with coefficient functions which have a dependence on the fermionic coordinates only through the CR holomorphic combinations $\eta_i$ and $\theta^i$. This graded algebra is augmented to a differential graded algebra via the differential 
        \begin{equation}\label{eq:CRantiholomorphicDifferentialTwistedReduced}
            \bar\partial_{\rm CR,\,tw,\,red}\ \coloneqq\ \hat v^\rmF\hat V_\rmF+\hat v^\rmL\hat V_\rmL+\hat v^\rmR\hat V_\rmR~,
        \end{equation}
        and we arrive at a third differential graded Lie algebra
        \begin{equation}\label{eq:SYMDGLATwistedReduced}
            \frL_{\rm CR,\,tw,\,red}\ \coloneqq\ \big(\Omega^{0,\bullet}_{\rm CR,\,tw,\,red}(F)\otimes\frg,\bar\partial_{\rm CR,\,tw,\,red},[-,-]\big)\,.
        \end{equation} 
        
        \begin{proposition}\label{prop:WittenGauge}
            The CR holomorphic Chern--Simons equations of motion defined by the differential graded Lie algebras $\frL_{\rm CR,\,tw}$ and $\frL_{\rm CR,\,tw,\,red}$ from~\eqref{eq:SYMDGLATwisted} and~\eqref{eq:SYMDGLATwistedReduced} (in their Batalin--Vilkovisky formulation including ghosts and anti-fields)\footnote{See discussion following~\eqref{eq:SYMDGLA}.} are equivalent, that is, $\frL_{\rm CR,\,tw}$ and $\frL_{\rm CR,\,tw,\,red}$ are quasi-isomorphic.
        \end{proposition}
        
        \begin{proof}
            The proof requires some preliminary considerations and is postponed to \cref{app:homologicalPerturbations}. 
        \end{proof}
        
        From a physical perspective, and for physical on-shell fields, the result is not surprising: it merely means that we can impose the \uline{Witten gauge}
        \begin{equation}\label{eq:WittenGauge}
            \hat V^i\intprod A\ =\ 0\ =\ \hat V_i\intprod A
        \end{equation}
        with $\hat V^i$ and $\hat V_i$ the (commuting) fermionic vector fields. In this gauge, the remaining components of $A$ will depend on the fermionic coordinates only CR holomorphically via $\eta_i$ and $\theta^i$. This gauge is very familiar from holomorphic Chern--Simons theory on Penrose's twistor space~\cite{Witten:2003nn}.
        
        As we shall see next, the Witten gauge makes the twisted CR holomorphic equation of motion~\eqref{eq:CRCSEoMTwisted} variational provided, of course, the Lie algebra $\frg$ is a metric Lie algebra.
        
        \section{CR ambitwistor action and its space-time interpretation}
        
        In the following, we shall assume that the Lie algebra $\frg$ also comes with an (invariant) metric $\inner{-}{-}$.
        
        \subsection{CR ambitwistor action}
        
        \paragraph{CR holomorphic volume form.} 
        To write down a Chern--Simons-type action, we need to construct an appropriate volume form. To this end, we consider the twisted CR holomorphic differential form 
        \begin{equation}
            \omega_{\rm CR,\,tw}\ \coloneqq\ v^\rmF\wedge v^\rmW\wedge v^{\hat\rmW}\wedge v^\rmL\wedge v^\rmR\wedge\rmd\eta_1\wedge\rmd\eta_2\wedge\rmd\eta_3\wedge\rmd\theta^1\wedge\rmd\theta^2\wedge\rmd\theta^3~.
        \end{equation}
        Note that we may replace $\rmd\eta_i$ by $v_i$ and $\rmd\theta^i$ by $v^i$ because of the appearance of $v^\rmL$ and $v^\rmR$ and as follows from~\eqref{eq:CRFermionicCoordinatesTransformations}.\footnote{See also \cref{fn:changeOfFermionicCoordinates} on page \cpageref{fn:changeOfFermionicCoordinates}.} For the same reason, the terms proportional to $\tilde\theta^i\tilde\eta_i$ in $v^\rmW$ and $v^{\hat\rmW}$ drop out. Hence, $\omega_{\rm CR,\,tw}$ depends on the fermionic coordinates only via the CR holomorphic combinations $\eta_i$ and $\theta^i$. This allows use now to transition to integral forms by requiring the Berezin integrations\footnote{By a slight abuse of notation, we shall use the same symbol to denote the integral form corresponding to a differential form. }
        \begin{equation}\label{eq:BerezinIntegration}
            \int\rmd\eta_i\eta_j\ \coloneqq\ \delta_{ij}
            \eand
            \int\rmd\theta^i\theta^j\ \coloneqq\ \delta^{ij}
        \end{equation}
        for the CR holomorphic coordinates. Consequently, we arrive at the integral form
        \begin{equation}\label{eq:twistedCRHolomorphicVolumeForm}
            \omega_{\rm CR,\,tw}\ \mapsto\ \Omega_{\rm CR,\,tw}\ \coloneqq\ v^\rmF\wedge v^\rmW\wedge v^{\hat\rmW}\wedge v^\rmL\wedge v^\rmR\otimes\underbrace{\rmd\eta_1\rmd\eta_2\rmd\eta_3}_{\eqqcolon\,\rmd^3\eta}\underbrace{\rmd\theta^1\rmd\theta^2\rmd\theta^3}_{\eqqcolon\,\rmd^3\theta}
        \end{equation}
        which we call the \uline{twisted CR holomorphic volume form}. It is now not too difficult to see that because of~\eqref{eq:BerezinIntegration}, $\Omega_{\rm CR,\,tw}$ is of homogeneity zero and thus, globally well defined on $F$. It is also $\bar\partial_{\rm CR,\,tw}$-closed which is a direct consequence of the commutation relations amongst the vector fields~\eqref{eq:tangentBundleDistributionsCRSuperAmbitwistorSpaceTwisted} and~\eqref{eq:tangentBundleDistributionsCRSuperAmbitwistorSpaceTwistedHolomorphic}.
        
        \paragraph{CR holomorphic Chern--Simons form and action.} 
        Furthermore, the terms proportional to $\theta^i\eta_i$ in $v^\rmW$ and $v^{\hat\rmW}$ in~\eqref{eq:twistedCRHolomorphicVolumeForm} also drop out in the wedge product of~\eqref{eq:twistedCRHolomorphicVolumeForm} with the \uline{twisted CR holomorphic Chern--Simons form}
        \begin{subequations}
            \begin{equation}\label{eq:twistedCRCSForm}
                {\rm CS}_{\rm CR,\,tw}\ \coloneqq\ \tfrac12\inner{A}{\bar\partial_{\rm CR,\,tw}A}+\tfrac{1}{3!}\inner{A}{[A,A]}~,
            \end{equation}
            where now $A\in\Omega^{0,1}_{\rm CR,\,tw,\,red}(F)\otimes\frg$ is taken to be in the Witten gauge~\eqref{eq:WittenGauge} the latter of which we can always assume by virtue of \cref{prop:WittenGauge}. Therefore, in the Witten gauge, the twisted CR holomorphic Chern--Simons equation~\eqref{eq:CRCSEoMTwisted} follows upon varying the \uline{twisted CR holomorphic Chern--Simons action}\footnote{We note that this action is similar in spirit of the Chern--Simons actions discussed e.g.~in~\cite{Nair:1990aa,Nair:1991ab,Popov:2005uv,Costello:2013zra}.}
            \begin{equation}\label{eq:twistedCRCSAction}
                S\ \coloneqq\ \int\Omega_{\rm CR,\,tw}\wedge{\rm CS}_{\rm CR,\,tw}~.
            \end{equation}
        \end{subequations}
        
        \begin{remark}\label{rmk:CRAmbitwistorSpace}
            Note that the action~\eqref{eq:twistedCRCSAction} can be understood to live on the real $8|12$-dimensional submanifold $L\rightarrow\IR^4\times\IC P^1\times\IC P^1$ inside the augmented CR ambitwistor space $F$ with $L$ the pull-back of the fermionic holomorphic vector bundle $[\caO(1,0)\oplus\caO(0,1)]\otimes\IC^{0|3}\rightarrow\IC P^1\times\IC P^1$ to the body $\IR^4\times\IC P^1\times\IC P^1\rightarrow\IC P^1\times\IC P^1$ of $F$. We shall call $L$ the \uline{CR ambitwistor space}. 
            
            We have the commutative triangle of fibrations
            \begin{equation}
                \begin{tikzcd}
                    & F\ar[dl,"\pi_2",swap]\ar[dr,"\pi_1"] &
                    \\
                    L\ar[rr,"\pi_3"] & & \IR^4
                \end{tikzcd}
            \end{equation}
            where $\pi_1$ is the trivial projection $\IR^4\times\IC P^1\times\IC P^1\rightarrow\IR^4$, $\pi_2:(x^{\alpha\dot\alpha},\eta^{\dot\alpha}_i,\theta^{i\alpha},\lambda_{\dot\alpha},\mu_\alpha)\mapsto(x^{\alpha\dot\alpha},\eta_i,\theta^i,\lambda_{\dot\alpha},\mu_\alpha)$ with the tangent spaces of the fibres of $\pi_2$ spanned by the (commuting) fermionic vector fields $\hat V^i$ and $\hat V_i$, and $\pi_3$ is the concatenation of the bundle projection $L\rightarrow\IR^4\times\IC P^1\times\IC P^1$ and the trivial projection $\IR^4\times\IC P^1\times\IC P^1\rightarrow\IR^4$. 
            
            Consequently, the twisted CR holomorphic volume form in~\eqref{eq:twistedCRCSAction} can be taken to be
            \begin{subequations}\label{eq:reducedCRExpressions}
                \begin{equation}
                    \Omega_{\rm CR,\,tw}\ \rightarrow\ \Omega_{\rm CR,\,tw,\,red}\ \coloneqq\ e^\rmF\wedge e^\rmW\wedge e^{\hat\rmW}\wedge e^\rmL\wedge e^\rmR\otimes\underbrace{\rmd\eta_1\rmd\eta_2\rmd\eta_3}_{\eqqcolon\,\rmd^3\eta}\underbrace{\rmd\theta^1\rmd\theta^2\rmd\theta^3}_{\eqqcolon\,\rmd^3\theta}~,
                \end{equation}
                where now $e^\rmF$, $e^\rmW$, $e^{\hat\rmW}$, $e^\rmL$, and $e^\rmR$ are as in~\eqref{eq:coTangentBundleDistributionsCRAmbitwistorSpace}. Likewise, the differential $\bar\partial_{\rm CR,\,tw}$ in~\eqref{eq:twistedCRCSForm} can be taken to be the reduced differential given in~\eqref{eq:CRantiholomorphicDifferentialTwistedReduced}, that is,
                \begin{equation}
                    \bar\partial_{\rm CR,\,tw}\ \rightarrow\ \bar\partial_{\rm CR,\,tw,\,red}\ =\ \hat e^\rmF\hat E_\rmF+\hat e^\rmL(\hat E_\rmL+\theta^i\eta_iE_{\hat\rmW})+\hat e^\rmR(\hat E_\rmR+\theta^i\eta_iE_\rmW)~,
                \end{equation}
            \end{subequations}
            where now $\hat e^\rmF$, $\hat e^\rmL$, and $\hat e^\rmR$ are as in~\eqref{eq:coTangentBundleDistributionsCRAmbitwistorSpace}. By construction, $\Omega_{\rm CR,\,tw,\,red}$ is $\bar\partial_{\rm CR,\,tw,\,red}$-closed.
        \end{remark}
        
        \paragraph{Batalin--Vilkovisky action.}
        We may also consider the Batalin--Vilkovisky extension of the CR ambitwistor action~\eqref{eq:twistedCRCSAction}. This action is obtained by replacing the Lie-algebra valued $(0,1)$-form $A$ in the twisted CR holomorphic Chern--Simons form with a general element $\caA\in\Omega^{0,\bullet}_{\rm CR,\,tw,\,red}(F)\otimes\frg$. Such an element is of the form
        \begin{subequations}
            \begin{equation}\label{eq:CRCSBVFieldContent}
                \caA\ =\ C+A+A^++C^+
            \end{equation} 
            with $C\in\Omega^{0,0}_{\rm CR,\,tw,\,red}(F)\otimes\frg$ the ghost, $A^+\in\Omega^{0,2}_{\rm CR,\,tw,\,red}(F)\otimes\frg$ the anti-field of the gauge potential $A$, and $C^+\in\Omega^{0,3}_{\rm CR,\,tw,\,red}(F)\otimes\frg$ the anti-field of $C$. In terms of these component fields, the corresponding Batalin--Vilkovisky action reads as
            \begin{equation}\label{eq:twistedCRCSActionBV}
                \begin{aligned}
                    S_{\rm BV}\ &\coloneqq\ \int\Omega_{\rm CR,\,tw}\wedge\Big\{\tfrac12\inner{A}{\bar\partial_{\rm CR,\,tw}A}+\tfrac{1}{3!}\inner{A}{[A,A]}
                    \\
                    &\kern3cm\ +\inner{A^+}{\bar\partial_{\rm CR,\,tw}C+[A,C]}+\tfrac12\inner{C^+}{[C,C]}\Big\}\,.
                \end{aligned}
            \end{equation}
        \end{subequations}
        
        \subsection{Equivalence to \texorpdfstring{$\caN=3$}{N=3} supersymmetric Yang--Mills theory}\label{ssec:equivalence}
        
        \Cref{prop:quasi-iso-first-cohomologies} shows that the twisted CR holomorphic Chern--Simons equation of motion and the equations of motion of $\caN=3$ supersymmetric Yang--Mills theory are equivalent. In the remainder of this paper, we shall demonstrate that this equivalence extends to the level of the ambitwistor action~\eqref{eq:twistedCRCSAction} and its Batalin--Vilkovisky extension~\eqref{eq:twistedCRCSActionBV}, culminating in \cref{thm:equivalence} below. Equivalence here is again precisely the semi-classical equivalence mentioned above, i.e.~$L_\infty$-quasi-isomorphy. Concretely, the differential graded Lie algebra defined by the CR holomorphic Chern--Simons action~\eqref{eq:twistedCRCSActionBV} and the differential graded Lie algebra defined by the first-order Batalin--Vilkovisky action~\eqref{eq:firstOrderYMActionBV} are quasi-isomorphic.
        
        Furthermore, this $L_\infty$-quasi-isomorphism can be phrased as a homotopy transfer and thus, physically, corresponds to integrating out infinitely many auxiliary fields in the action~\eqref{eq:twistedCRCSActionBV}. For the reader's convenience, we summarise the key formulas about homotopy transfer in~\cref{app:homologicalPerturbations}.
        
        The proof of the corresponding theorems, \cref{thm:equivalence} and \cref{thm:homotopy-transfer},  is broken down into several steps. Firstly, we give a brief review of the first-order formulation of Yang--Mills theory and its Batalin--Vilkovisky extension. Secondly, we establish that there is a cyclic quasi-isomorphism between the cochain complexes underlying both theories. Finally, we show that this quasi-isomorphism extends to an $L_\infty$-quasi-isomorphism between the differential graded Lie algebras governing both theories and that this quasi-isomorphism is a homotopy transfer. To keep the formulas manageable, we shall restrict the explicit parts of our calculations to the R-symmetry singlets of the $\caN=3$ multiplet, that is, to the gluons. The full equivalence follows then from covariance of all our constructions under supersymmetry.
        
        \paragraph{First-order Yang--Mills action.} 
        As before, let $\frg$ be a Lie algebra with Lie bracket $[-,-]$ and inner product $\inner{-}{-}$. The standard second-order Yang--Mills action is
        \begin{subequations}
            \begin{equation}\label{eq:secondOrderYMAction}
                S\ =\ \tfrac12\int\rmd^4x\,\Big\{\inner{f^{\dot\alpha\dot\beta}}{f_{\dot\alpha\dot\beta}}+\inner{f^{\alpha\beta}}{f_{\alpha\beta}}\Big\}\,,
            \end{equation}
            where
            \begin{equation}\label{eq:chiralCurvatures}
                \begin{aligned}
                    f_{\dot\alpha\dot\beta}\ &\coloneqq\ \tfrac12\eps^{\alpha\beta}\big\{\partial_{\alpha\dot\alpha}A_{\beta\dot\beta}-\partial_{\beta\dot\beta}A_{\alpha\dot\alpha}+[A_{\alpha\dot\alpha},A_{\beta\dot\beta}]\big\}\ =\ \eps^{\alpha\beta}\partial_{\alpha(\dot\alpha}A_{\beta\dot\beta)}+\tfrac12\eps^{\alpha\beta}[A_{\alpha\dot\alpha},A_{\beta\dot\beta}]\,,
                    \\
                    f_{\alpha\beta}\ &\coloneqq\ \tfrac12\eps^{\dot\alpha\dot\beta}\big\{\partial_{\alpha\dot\alpha}A_{\beta\dot\beta}-\partial_{\beta\dot\beta}A_{\alpha\dot\alpha}+[A_{\alpha\dot\alpha},A_{\beta\dot\beta}]\big\}\ =\ \eps^{\dot\alpha\dot\beta}\partial_{(\alpha\dot\alpha}A_{\beta)\dot\beta}+\tfrac12\eps^{\dot\alpha\dot\beta}[A_{\alpha\dot\alpha},A_{\beta\dot\beta}]
                \end{aligned}
            \end{equation}
        \end{subequations}
        are the anti-self-dual and self-dual parts of the curvature of the Lie-algebra-valued one-form $A_{\alpha\dot\alpha}$. Introducing the anti-self-dual, $B_{\dot\alpha\dot\beta}=B_{\dot\beta\dot\alpha}$, and self-dual, $B_{\alpha\beta}=B_{\beta\alpha}$, parts of an auxiliary Lie-algebra valued auxiliary two-form transforming in the adjoint representation of the gauge group, we can write the first-order action of Yang--Mills theory as 
        \begin{equation}\label{eq:firstOrderYMAction}
            S\ =\ \int\rmd^4x\,\Big\{\inner{B^{\dot\alpha\dot\beta}}{f_{\dot\alpha\dot\beta}-\tfrac12B_{\dot\alpha\dot\beta}}+\inner{B^{\alpha\beta}}{f_{\alpha\beta}-\tfrac12B_{\alpha\beta}}\Big\}\,.
        \end{equation}
        Evidently, upon integrating out $B_{\dot\alpha\dot\beta}$ and $B_{\alpha\beta}$, we recover the second-order Yang--Mills action~\eqref{eq:secondOrderYMAction}. The equations of motion following from~\eqref{eq:firstOrderYMAction} read as
        \begin{subequations}\label{eq:firstOrderYMEOM}
            \begin{equation}\label{eq:firstOrderYMEOMStandard}
                f_{\dot\alpha\dot\beta}\ =\ B_{\dot\alpha\dot\beta}~,~~~
                f_{\alpha\beta}\ =\ B_{\alpha\beta}~,
                \eand
                \eps^{\dot\beta\dot\gamma}\nabla_{\alpha\dot\beta}B_{\dot\gamma\dot\alpha}+\eps^{\beta\gamma}\nabla_{\beta\dot\alpha}B_{\gamma\alpha}\ =\ 0~,
            \end{equation}
            where, as before, $\nabla_{\alpha\dot\alpha}=\partial_{\alpha\dot\alpha}+[A_{\alpha\dot\alpha},-]$. Because of the Bianchi identity, 
            \begin{equation}
                \eps^{\dot\beta\dot\gamma}\nabla_{\alpha\dot\beta}f_{\dot\gamma\dot\alpha}-\eps^{\beta\gamma}\nabla_{\beta\dot\alpha}f_{\gamma\alpha}\ =\ 0~,
            \end{equation}
            the equations of motion~\eqref{eq:firstOrderYMEOMStandard} are equivalent to 
            \begin{equation}
                f_{\dot\alpha\dot\beta}\ =\ B_{\dot\alpha\dot\beta}~,~~~
                f_{\alpha\beta}\ =\ B_{\alpha\beta}~,~~~
                \eps^{\dot\beta\dot\gamma}\nabla_{\alpha\dot\beta}B_{\dot\gamma\dot\alpha}\ =\ 0~,
                \eand
                \eps^{\beta\gamma}\nabla_{\beta\dot\alpha}B_{\gamma\alpha}\ =\ 0~.
            \end{equation}
        \end{subequations}
        
        Furthermore, the Batalin--Vilkovisky extension of the first-order Yang--Mills action~\eqref{eq:firstOrderYMAction} reads as
        \begin{equation}\label{eq:firstOrderYMActionBV}
            \begin{aligned}
                S_{\rm BV}\ &=\ \int\rmd^4x\,\Big\{\inner{B^{\dot\alpha\dot\beta}}{f_{\dot\alpha\dot\beta}-\tfrac12B_{\dot\alpha\dot\beta}}+\inner{B^{\alpha\beta}}{f_{\alpha\beta}-\tfrac12B_{\alpha\beta}}
                \\
                &\kern1.5cm+\inner{A^+_{\alpha\dot\alpha}}{\nabla^{\alpha\dot\alpha}c}+\inner{B^+_{\dot\alpha\dot\beta}}{[B^{\dot\alpha\dot\beta},c]}+\inner{B^+_{\alpha\beta}}{[B^{\alpha\beta},c]}+\tfrac12\inner{c^+}{[c,c]}\Big\}\,,
            \end{aligned}
        \end{equation}
        where $c$ is the ghost field, and $A^+_{\alpha\dot\alpha}$, $B^+_{\dot\alpha\dot\beta}=B^+_{\dot\beta\dot\alpha}$, $B^+_{\alpha\beta}=B^+_{\beta\alpha}$, and $c^+$ are the evident anti-fields. 
        
        As familiar from the example of Chern--Simons theory, also this action defines a differential graded Lie algebra $\frL_{\rm YM_1}$ with the underlying cochain complex
        \begin{subequations}\label{eq:DGLAFirstOrderYM}
            \begin{equation} 
                \begin{tikzcd}[row sep=-0.1cm]
                    \stackrel{c}{\Omega^0(\IR^4)\otimes\frg}\arrow[r,"\partial_{\alpha\dot\alpha}"]
                    & 
                    \stackrel{A_{\alpha\dot\alpha}}{\Omega^1(\IR^4)\otimes\frg}\arrow[start anchor=east,end anchor=west,rdd,"\partial_{\alpha\dot\alpha}",pos=0.2]
                    & [20pt]
                    \stackrel{A^+_{\alpha\dot\alpha}}{\Omega^1(\IR^4)\otimes\frg}\arrow[r,"\partial^{\alpha\dot\alpha}"]
                    &
                    \stackrel{c^+}{\Omega^0(\IR^4)\otimes\frg}
                    \\
                    & \oplus & \oplus
                    \\
                    \underbrace{\spacer{2ex}\phantom{\Omega^0(\IR^4,\frg)}}_{\eqqcolon\,\frL_{\rm YM_1,\,0}} & 
                    \underbrace{\spacer{2ex}\stackrel{B_{\dot\alpha\dot\beta},\,B_{\alpha\beta}}{\Omega^2(\IR^4)\otimes\frg}}_{\eqqcolon\,\frL_{\rm YM_1,\,1}}\arrow[start anchor=east,end anchor=west,uur,"\partial^{\alpha\dot\alpha}",pos=0.2,crossing over]\arrow[r,"-\sfid",swap,shift right]
                    &
                    \underbrace{\spacer{2ex}\stackrel{B^+_{\dot\alpha\dot\beta},\,B^+_{\alpha\beta}}{\Omega^2(\IR^4)\otimes\frg}}_{\eqqcolon\,\frL_{\rm YM_1,\,2}}
                    &
                    \underbrace{\spacer{2ex}\phantom{\Omega^0(\IR^4,\frg)}}_{\eqqcolon\,\frL_{\rm YM_1,\,3}}
                \end{tikzcd}
            \end{equation}
            and the binary products $\mu_2$ defined by
            \begin{equation}
                \begin{gathered}
                    \mu_2(c_1,c_2)\ \coloneqq\ [c_1,c_2]~,~~~
                    \mu_2(c,A_{\alpha\dot\alpha})\ \coloneqq\ [c,A_{\alpha\dot\alpha}]~,
                    \\
                    \mu_2(c,B_{\dot\alpha\dot\beta})\ \coloneqq\ [c,B_{\dot\alpha\dot\beta}]~,~~~
                    \mu_2(c,B_{\alpha\beta})\ \coloneqq\ [c,B_{\alpha\beta}]~,
                    \\
                    \mu_2(c,c^+)\ \coloneqq\ [c,c^+]~,~~~
                    \mu_2(c,A^+_{\alpha\dot\alpha})\ \coloneqq\ -[c,A^+_{\alpha\dot\alpha}]~,
                    \\
                    \mu_2(c,B^+_{\dot\alpha\dot\beta})\ \coloneqq\ -[c,B^+_{\dot\alpha\dot\beta}]~,~~~
                    \mu_2(c,B^+_{\alpha\beta})\ \coloneqq\ -[c,B^+_{\alpha\beta}]~,
                    \\
                    \mu_2(A_{1\,\alpha\dot\alpha},A_{2\,\beta\dot\beta})\ \coloneqq\ \tfrac12\big(\eps^{\alpha\beta}[A_{1\,\alpha(\dot\alpha},A_{2\,\beta\dot\beta)}],\eps^{\dot\alpha\dot\beta}[A_{1\,(\alpha\dot\alpha},A_{2\,\beta)\dot\beta}]\big)\,,
                    \\
                    \mu_2(A_{\alpha\dot\alpha},B_{\dot\beta\dot\gamma})\ \coloneqq\ \eps^{\dot\beta\dot\gamma}[A_{\alpha\dot\beta},B_{\dot\gamma\dot\alpha}]~,~~~
                    \mu_2(A_{\alpha\dot\alpha},B_{\beta\gamma})\ \coloneqq\ \eps^{\beta\gamma}[A_{\beta\dot\alpha},B_{\gamma\alpha}]~,
                    \\
                    \mu_2(A_{\alpha\dot\alpha},A^+_{\beta\dot\beta})\ \coloneqq\ -[A^{\alpha\dot\alpha},A^+_{\alpha\dot\alpha}]~,
                    \\
                    \mu_2(B_{\dot\alpha\dot\beta},B^+_{\dot\gamma\dot\delta})\ \coloneqq\ -[B^{\dot\alpha\dot\beta},B^+_{\dot\alpha\dot\beta}]~,~~~
                    \mu_2(B_{\alpha\beta},B^+_{\gamma\delta})\ \coloneqq\ -[B^{\alpha\beta},B^+_{\alpha\beta}]~.
                \end{gathered}
            \end{equation}
        \end{subequations}

        One can show that this first-order formulation is semi-classically equivalent to the second-order formulation following the constructions in~\cite{Costello:2011aa,Rocek:2017xsj,Jurco:2018sby}. The same applies to the $\caN=3$ supersymmetric extension.
        
        \paragraph{Embedding of theories.}
        Consider the differential graded Lie algebras $\frL_{\rm CR,\,tw,\,red}$ and $\frL_{\rm YM_1}$ defined in~\eqref{eq:SYMDGLATwistedReduced} and~\eqref{eq:DGLAFirstOrderYM}, respectively, and define a degree-zero morphism of graded vector spaces 
        \begin{subequations}\label{eq:spaceTimeAnsatzBV}
            \begin{equation}
                \begin{gathered}
                    \sfe\,:\,\frL_{\rm YM_1}\ \rightarrow\ \frL_{\rm CR,\,tw,\,red}~,
                    \\
                    c\ \mapsto\ C~,~~~
                    \binom{A_{\alpha\dot\alpha}}{B_{\dot\alpha\dot\beta},B_{\alpha\beta}}\ \mapsto\ A~,~~~
                    \binom{A^+_{\alpha\dot\alpha}}{B^+_{\dot\alpha\dot\beta},B^+_{\alpha\beta}}\ \mapsto\ A^+~,~~
                    c^+\ \mapsto\ C^+
                \end{gathered}
            \end{equation}
            between the field space of first-order Yang--Mills theory and the field space of CR holomorphic Chern--Simons theory by setting 
            \begin{equation}
                C\ \coloneqq\ c
                \eand
                C^+\ \coloneqq\ \hat v^\rmF\wedge\hat v^\rmL\wedge\hat v^\rmR\,(\theta^i\eta_i)^3\frac43c^+~, 
            \end{equation}
            \begin{equation}
                \begin{aligned}
                    A\ &\coloneqq\ \hat v^\rmF\left\{A_{\alpha\dot\alpha}\mu^\alpha\lambda^{\dot\alpha}-\theta^i\eta_i\left(\frac{B_{\dot\alpha\dot\beta}\lambda^{\dot\alpha}\hat\lambda^{\dot\beta}}{|\lambda|^2}-\frac{B_{\alpha\beta}\mu^\alpha\hat\mu^\beta}{|\mu|^2}\right)\right.
                    \\
                    &\kern1.5cm-(\theta^i\eta_i)^2\left(\frac{\partial_{\alpha(\dot\alpha}B_{\dot\beta\dot\gamma)}\hat\mu^\alpha\lambda^{\dot\alpha}\hat\lambda^{\dot\beta}\hat\lambda^{\dot\gamma}}{2|\mu|^2|\lambda|^4}+\frac{\partial_{(\alpha\dot\alpha}B_{\beta\gamma)}\mu^\alpha\hat\mu^\beta\hat\mu^\gamma\hat\lambda^{\dot\alpha}}{2|\mu|^4|\lambda|^2}\right)
                    \\
                    &\kern1.5cm\left.-(\theta^i\eta_i)^3\left(\frac{\partial_{\alpha(\dot\alpha}\partial_{\beta\dot\beta}B_{\dot\gamma\dot\delta)}\hat\mu^\alpha\hat\mu^\beta\lambda^{\dot\alpha}\hat\lambda^{\dot\beta}\hat\lambda^{\dot\gamma}\hat\lambda^{\dot\delta}}{6|\mu|^4|\lambda|^6}-\frac{\partial_{(\alpha\dot\alpha}\partial_{\beta\dot\beta}B_{\gamma\delta)}\mu^\alpha\hat\mu^\beta\hat\mu^\gamma\hat\mu^\delta\hat\lambda^{\dot\alpha}\hat\lambda^{\dot\beta}}{6|\mu|^6|\lambda|^4}\right)\right\}
                    \\
                    &\kern1cm+\hat v^\rmL\left\{-\theta^i\eta_i\frac{A_{\alpha\dot\alpha}\hat\mu^\alpha\lambda^{\dot\alpha}}{|\mu|^2}-(\theta^i\eta_i)^2\frac{3B_{\alpha\beta}\hat\mu^\alpha\hat\mu^\beta}{2|\mu|^4}+(\theta^i\eta_i)^3\frac{5\partial_{(\alpha\dot\alpha}B_{\beta\gamma)}\hat\mu^\alpha\hat\mu^\beta\hat\mu^\gamma\hat\lambda^{\dot\alpha}}{6|\mu|^6|\lambda|^2}\right\}
                    \\
                    &\kern1cm+\hat v^\rmR\left\{\theta^i\eta_i\frac{A_{\alpha\dot\alpha}\mu^\alpha\hat\lambda^{\dot\alpha}}{|\lambda|^2}-(\theta^i\eta_i)^2\frac{3B_{\dot\alpha\dot\beta}\hat\lambda^{\dot\alpha}\hat\lambda^{\dot\beta}}{2|\lambda|^4}-(\theta^i\eta_i)^3\frac{5\partial_{\alpha(\dot\alpha}B_{\dot\beta\dot\gamma)}\hat\mu^\alpha\hat\lambda^{\dot\alpha}\hat\lambda^{\dot\beta}\hat\lambda^{\dot\gamma}}{6|\mu|^2|\lambda|^6}\right\},
                \end{aligned}
            \end{equation}
            and 
            \begin{equation}
                \begin{aligned}
                    A^+\ &\coloneqq\ \hat v^\rmL\wedge\hat v^\rmR\left\{(\theta^i\eta_i)^2\left(\frac{B^+_{\dot\alpha\dot\beta}\lambda^{\dot\alpha}\hat\lambda^{\dot\beta}}{|\lambda|^2}-\frac{B^+_{\alpha\beta}\mu^\alpha\hat\mu^\beta}{|\mu|^2}\right)+(\theta^i\eta_i)^3\frac{A^+_{\alpha\dot\alpha}\hat\mu^\alpha\hat\lambda^{\dot\alpha}}{|\mu|^2|\lambda|^2}\right\}
                    \\
                    &\kern1cm+\hat v^\rmF\wedge\hat v^\rmL\left\{-\theta^i\eta_iB^+_{\dot\alpha\dot\beta}\lambda^{\dot\alpha}\lambda^{\dot\beta}+(\theta^i\eta_i)^2\frac{\hat\mu^\alpha\lambda^{\dot\alpha}}{|\mu|^2}\left(-\frac56A^+_{\alpha\dot\alpha}-\frac12B^+_{\alpha\dot\alpha}\right)\right.
                    \\
                    &\kern1cm\left.+(\theta^i\eta_i)^3\frac{\hat\mu^\alpha\hat\mu^\beta\lambda^{\dot\alpha}\hat\lambda^{\dot\beta}}{|\mu|^4|\lambda|^2}\left(\frac14\eps_{\dot\alpha\dot\beta}\eps^{\dot\gamma\dot\delta}\partial_{(\alpha\dot\gamma}\big(A^+_{\beta)\dot\delta}+B^+_{\beta)\dot\delta}\big)+\frac16\partial_{(\alpha(\dot\alpha}\big(A^+_{\beta)\dot\beta)}+2B^+_{\beta)\dot\beta)}\big)\right)\right\}
                    \\
                    &\kern1cm+\hat v^\rmF\wedge\hat v^\rmR\left\{\theta^i\eta_iB^+_{\alpha\beta}\mu^\alpha\mu^\beta+(\theta^i\eta_i)^2\frac{\mu^\alpha\hat\lambda^{\dot\alpha}}{|\lambda|^2}\left(-\frac56A^+_{\alpha\dot\alpha}+\frac12B^+_{\alpha\dot\alpha}\right)\right.
                    \\
                    &\kern1cm\left.-(\theta^i\eta_i)^3\frac{\mu^\alpha\hat\mu^\beta\hat\lambda^{\dot\alpha}\hat\lambda^{\dot\beta}}{|\mu|^2|\lambda|^4}\left(\frac14\eps_{\alpha\beta}\eps^{\gamma\delta}\partial_{\gamma(\dot\alpha}\big(A^+_{\delta\dot\beta)}-B^+_{\delta\dot\beta)}\big)+\frac16\partial_{(\alpha(\dot\alpha}\big(A^+_{\beta)\dot\beta)}-2B^+_{\beta)\dot\beta)}\big)\right)\right\},
                \end{aligned}
            \end{equation}
            where we have used the short-hand notation
            \begin{equation}
                B^+_{\alpha\dot\alpha}\ \coloneqq\ \eps^{\dot\beta\dot\gamma}\partial_{\alpha\dot\beta}B^+_{\dot\gamma\dot\alpha}-\eps^{\beta\gamma}\partial_{\beta\dot\alpha}B^+_{\gamma\alpha}~.
            \end{equation}
        \end{subequations}
        
        \begin{proposition}\label{eq:cyclicityPreservingProperty}
            Restricted to the image of $\sfe$ defined in~\eqref{eq:spaceTimeAnsatzBV}, the action~\eqref{eq:twistedCRCSActionBV} reduces to the action~\eqref{eq:firstOrderYMActionBV}.
        \end{proposition}

        \begin{proof}
            The proof follows from a straightforward but lengthy computation. We briefly illustrate the computation using the classical part of the action~\eqref{eq:twistedCRCSAction}. Firstly, one can check that the derivative terms appearing in the expression of $A$ in~\eqref{eq:spaceTimeAnsatzBV} will not contribute as we are only interested in terms of order $(\theta^i\eta_i)^3$ when computing the action. Next, upon writing 
            \begin{equation}
				\begin{aligned}
                    A\ &=\ \hat v^\rmF A_\rmF+\hat v^\rmL A_\rmL+\hat v^\rmR A_\rmR
                    \\
                    \ &=\ \hat v^\rmF\sum_n(\theta^i\eta_i)^nA^{(n)}_\rmF+\hat v^\rmL\sum_n(\theta^i\eta_i)^nA^{(n)}_\rmL+\hat v^\rmR\sum_n(\theta^i\eta_i)^nA^{(n)}_\rmR~,
				\end{aligned}
			\end{equation}
            a straightforward calculation shows that\footnote{See also~\eqref{eq:reducedCRExpressions}.}
            \begin{equation}\label{eq:kinematicTerm}
                \begin{aligned}
                    &\tfrac12\inner{A}{\bar\partial_{\rm CR,\,tw}A}\Big|_{(\theta^i\eta_i)^3}
                    \\ 
                    &=\ \hat e^\rmF\wedge\hat e^\rmL\wedge\hat e^\rmR\Big\{\inner{A_\rmF^{(1)}}{E_{\hat\rmW}A_\rmR^{(1)}-E_\rmW A_\rmL^{(1)}+\hat E_\rmL A_\rmR^{(2)}-\hat E_\rmR A_\rmL^{(2)}-A_\rmF^{(1)}}
                    \\
                    &\kern3.5cm-\inner{A_\rmL^{(2)}}{\hat E_\rmF A_\rmR^{(1)}-E_\rmW A_\rmF^{(0)}}+\inner{A_\rmR^{(2)}}{\hat E_\rmF A_\rmL^{(1)}-E_{\hat\rmW}A_\rmF^{(0)}}\Big\}
                \end{aligned}
            \end{equation}
            with the differential forms and vector fields as given in~\eqref{eq:tangentBundleDistributionsCRAmbitwistorSpace} and~\eqref{eq:coTangentBundleDistributionsCRAmbitwistorSpace}. Because $e^\rmF\wedge\hat e^\rmF\wedge e^\rmW\wedge e^{\hat\rmW}$ is the volume form $\rmd^4x$ on $\IR^4$ up to a multiplicative constant and because of the identities\footnote{Here, $\rmi\coloneqq\sqrt{-1}$. For details, see e.g.~\cite{Wolf:2010av}.}
            \begin{equation}\label{eq:SerreDuality}
                \begin{aligned}
                    -\frac{1}{2\pi\rmi}\int e^\rmL\wedge\hat e^\rmL\,f_{\dot\alpha_1\cdots\dot\alpha_m}g^{\dot\beta_1\cdots\dot\beta_m}\frac{\lambda_{\dot\beta_1}\cdots\lambda_{\dot\beta_m}\hat\lambda^{\dot\alpha_1}\cdots\hat\lambda^{\dot\alpha_m}}{|\lambda|^{2m}}\ &=\ \frac{1}{m+1}f_{\dot\alpha_1\cdots\dot\alpha_m}g^{\dot\alpha_1\cdots\dot\alpha_m}~,
                    \\
                    -\frac{1}{2\pi\rmi}\int e^\rmR\wedge\hat e^\rmR\,f_{\alpha_1\cdots\alpha_n}g^{\beta_1\cdots\beta_n}\frac{\mu_{\beta_1}\cdots\mu_{\beta_n}\hat\mu^{\alpha_1}\cdots\hat\mu^{\alpha_n}}{|\mu|^{2n}}\ &=\ \frac{1}{n+1}f_{\alpha_1\cdots\alpha_n}g^{\alpha_1\cdots\alpha_n}~,
                \end{aligned}
            \end{equation}
            the kinematic term~\eqref{eq:kinematicTerm} becomes 
            \begin{subequations}\label{eq:twoLeadingTermsYMAction}
                \begin{equation}
                    \begin{aligned}
                        &\int\Omega_{\rm CR,\,tw}\wedge\tfrac12\inner{A}{\bar\partial_{\rm CR,\,tw}A}
                        \\
                        &\kern1cm=\ \int\rmd^4x\,\Big\{\inner{B^{\dot\alpha\dot\beta}}{\eps^{\alpha\beta}\partial_{\alpha(\dot\alpha}A_{\beta\dot\beta)}-\tfrac12B_{\dot\alpha\dot\beta}}+\inner{B^{\alpha\beta}}{\eps^{\dot\alpha\dot\beta}\partial_{(\alpha\dot\alpha}A_{\beta)\dot\beta}-\tfrac12B_{\alpha\beta}}\Big\}
                    \end{aligned}
                \end{equation}
                up to an overall multiplicative constant. Likewise, the interaction term becomes
                \begin{equation}
                    \begin{aligned}
                        &\int\Omega_{\rm CR,\,tw}\wedge\tfrac{1}{3!}\inner{A}{[A,A]}\ =\ \int\rmd^4x\,\Big\{\inner{B^{\dot\alpha\dot\beta}}{\tfrac12\eps^{\alpha\beta}[A_{\alpha\dot\alpha},A_{\beta\dot\beta}]}+\inner{B^{\alpha\beta}}{\tfrac12\eps^{\dot\alpha\dot\beta}[A_{\alpha\dot\alpha},A_{\beta\dot\beta}]}\Big\}
                    \end{aligned}
                \end{equation}
            \end{subequations}
            up to the same overall multiplicative constant. Upon combining the two expressions in~\eqref{eq:twoLeadingTermsYMAction}, we conclude that to leading order, the ambitwistor action~\eqref{eq:twistedCRCSAction} becomes the first-order Yang--Mills action~\eqref{eq:firstOrderYMAction}, up to an overall multiplicative constant.
        \end{proof}

        \noindent 
        Physically speaking, we may conclude that the action~\eqref{eq:twistedCRCSActionBV} contains Yang--Mills theory in its first-order formulation, and because our formalism is fully covariant under $\caN=3$ supersymmetry, it actually contains $\caN=3$ supersymmetric Yang--Mills theory. It remains to show that the infinite tower of extra fields not contained in the image of $\sfe$ can be consistently integrated out.

        \paragraph{Quasi-isomorphism of cochain complexes.}
        Let us now tighten the relation between both actions.

        \begin{proposition}\label{prop:injectiveCochainMorphism}
            The morphism of graded vector spaces defined in~\eqref{eq:spaceTimeAnsatzBV} is an injective morphism of cochain complexes. This holds also for the $\caN=3$ supersymmetric extension.
        \end{proposition}
        
        \begin{proof}
            It is evident that $\sfe$ is injective. To check that it is a cochain map, we have to verify that $\bar\partial_{\rm CR,\,tw,\,red}\circ\sfe=\sfe\circ\mu_1$ where $\mu_1$ is the differential of the cochain complex in~\eqref{eq:DGLAFirstOrderYM} and $\bar\partial_{\rm CR,\,tw,\,red}$ was defined in~\eqref{eq:CRantiholomorphicDifferentialTwistedReduced}. This, however, follows from a direct, albeit lengthy, calculation.
        \end{proof}
        
        \begin{proposition}\label{prop:quasiIsaCochainMorphism}
            The $\caN=3$ supersymmetric extension of the morphism of cochain complexes defined in~\eqref{eq:spaceTimeAnsatzBV} is a quasi-isomorphism of cochain complexes under the assumption of $\IR^4$-triviality.\footnote{See around~\eqref{eq:ambitwistorGauge} for the notion of $\IR^4$-triviality.}
        \end{proposition}
        
        \begin{proof}
            We note that $H^\bullet(\frL_{\rm CR,\,tw,\,red})\cong H^\bullet(\frL_{\rm CR})$ by \cref{prop:quasiIsoTwisted} and \cref{prop:WittenGauge}. By~\cref{prop:quasi-iso-first-cohomologies}, we have additionally that $H^p(\frL_{\rm CR})\cong H^p(\frL_{\rm YM_1})$ for $p=0,1$. Therefore,
            \begin{equation}\label{eq:isomorphies-to-show}
                H^p(\frL_{\rm CR,\,tw,\,red})\ \cong\ H^p(\frL_{\rm YM_1})
            \end{equation}
            for $p=0,1$. It now remains to show that~\eqref{eq:isomorphies-to-show} also holds for $p=2,3$ and that the $\caN=3$ supersymmetric extension of $\sfe$ descends to an isomorphism on cohomology. Both of these statements can be shown by a direct but lengthy computation. 
            
            Alternatively, we can invoke the existence\footnote{See e.g.~\cite{Chuang:0810.2393,Doubek:2017naz}.} of a Hodge--Kodaira decomposition\footnote{See \cref{app:homologicalPerturbations}.} compatible with the natural inner product structure on $\frL_{\rm YM_1}$ and $\frL_{\rm CR,\,tw,\,red}$, the latter of which pairs elements of degree~0 and~3 as well as elements of degree~1 and~2, respectively. This non-degenerate pairing descends to a non-degenerate pairing on the cohomologies, and we obtain the isomorphisms~\eqref{eq:isomorphies-to-show} for $p=2,3$ from those for $p=0,1$ in this manner.\footnote{Strictly speaking, this argument needs to be refined since we are working with infinite-dimensional vector spaces; see~\cref{rem:inf-dim} for further details on this point.}
            
            The fact that the $\caN=3$ supersymmetric extension of $\sfe$ descends to an isomorphism on cohomologies follows similarly. Evidently, it descends to an injection on the cohomologies but it remains to show that this is a surjection as well.\footnote{A priori, this is not clear, as the cohomology groups are infinite-dimensional vector spaces.} However, a direct computation shows surjectivity $H^p(\frL_{\rm YM_1})\rightarrow H^p(\frL_{\rm CR,\,tw,\,red})$ for $p=0,1$, and surjectivity for $p=2,3$ can again be gleaned from the cyclic Hodge--Kodaira decomposition.
        \end{proof}
        
        \paragraph{Equivalence theorems.}
        Let us now promote $\sfe$ to a full quasi-isomorphism of differential graded Lie algebras, which take the following form.\footnote{See~\cite{Jurco:2018sby} for our conventions of quasi-isomorphism of $L_\infty$-algebras.} Given two differential graded Lie algebras $\frL^{(1)}$ and $\frL^{(2)}$ with differentials $\mu^{(1)}_1$ and $\mu^{(2)}_1$ and binary products $\mu^{(1)}_2$ and $\mu^{(2)}_2$, a \uline{weak morphism}
        \begin{subequations}\label{eq:LinftyMorphism}
            \begin{equation}
                \sfE\,:\,\frL^{(1)}\ \rightarrow \ \frL^{(2)}
            \end{equation} 
            consists of a collection of maps $\sfE_i:\frL^{(1)}\times\cdots\times\frL^{(1)}\rightarrow\frL^{(2)}$ of degree $1-i$ for $i=1,2,3,\ldots$ which satisfy
            \begin{equation}
                \begin{aligned}
                    &\sum_{j+k=i}\sum_{\sigma\in\overline{\rm Sh}(j;i)}~(-1)^{k}\chi(\sigma;\ell_1,\ldots,\ell_i)\sfE_{k+1}\big(\mu^{(1)}_j(\ell_{\sigma(1)},\ldots,\ell_{\sigma(j)}),\ell_{\sigma(j+1)},\ldots,\ell_{\sigma(i)}\big)
                    \\
                    &=\ \sum_{j=1}^i\frac{1}{j!}\sum_{k_1+\cdots+k_j=i}\sum_{\sigma\in\overline{\rm Sh}(k_1,\ldots,k_{j-1};i)}\chi(\sigma;\ell_1,\ldots,\ell_i)\zeta(\sigma;\ell_1,\ldots,\ell_i)
                    \\
                    &\kern1cm\times\mu^{(2)}_j\big(\sfE_{k_1}\big(\ell_{\sigma(1)},\ldots,\ell_{\sigma(k_1)}\big),\ldots,\sfE_{k_j}\big(\ell_{\sigma(k_1+\cdots+k_{j-1}+1)},\ldots,\ell_{\sigma(i)}\big)\big)
                \end{aligned}
            \end{equation}
            for all $\ell_1,\ell_2,\ell_3,\ldots\in\frL^{(1)}$, where the sum is over unshuffles, $\chi(\sigma;\ell_1,\ldots,\ell_i)$ is the Koszul sign, and $\zeta(\sigma;\ell_1,\ldots,\ell_i)$ for a $(k_1,\ldots, k_{j-1};i)$-unshuffle $\sigma$ is defined as
            \begin{equation}
                \zeta(\sigma;\ell_1,\ldots,\ell_i)\ \coloneqq\ (-1)^{\sum_{1\leq m<n\leq j}k_mk_n+\sum_{m=1}^{j-1}k_m(j-m)+\sum_{m=2}^j(1-k_m)\sum_{k=1}^{k_1+\cdots+k_{m-1}}|\ell_{\sigma(k)}|}~.
            \end{equation}
        \end{subequations}
        A weak morphisms becomes a quasi-isomorphism if and only if $\sfE_1$ descends to an isomorphism on cohomology. Note that~\eqref{eq:LinftyMorphism} are the defining relations for a \uline{morphism of $L_\infty$-algebras} if we allow for higher products $\mu^{(1)}_i$ and $\mu^{(2)}_i$ with $i=3,4,\ldots$. See e.g.~\cite{Jurco:2018sby,Jurco:2019bvp} for details. 
        
        Returning to our differential graded Lie algebras $\frL_{\rm CR,\,tw,\,red}$ and $\frL_{\rm YM_1}$ defined in ~\eqref{eq:DGLAFirstOrderYM} and~\eqref{eq:SYMDGLATwistedReduced}, we take $\sfE$ with only $\sfE_1$, $\sfE_2$, and $\sfE_3$ non-zero and given by
        \begin{equation}\label{eq:DGLAMorphism}
            \sfE\ \coloneqq\ \sfe|_{\partial_{\alpha\dot\alpha}\to\nabla_{\alpha\dot\alpha}}
            \ewith
            \nabla_{\alpha\dot\alpha}\ =\ \partial_{\alpha\dot\alpha}+[A_{\alpha\dot\alpha},-]
        \end{equation}
        with $\sfe$ as defined in~\eqref{eq:spaceTimeAnsatzBV}.\footnote{Note that this morphism is reminiscent of the `on-shell' expansions~\cite{Popov:2004rb,Saemann:2004tt,Popov:2004nk,Popov:2005uv,Lechtenfeld:2005xi} in the self-dual sector, although these references only discuss gauge potentials.} Evidently, $\sfE_1=\sfe$. We have now all the ingredients to state and prove the first central result.\footnote{In the self-dual sector and in the fully complex setting, the Penrose--Ward transform was shown to be a quasi-isomorphism using spans of $L_\infty$-algebras and homotopy transfer in~\cite{JalaliFarahani:2023sfq}.}
        
        \begin{theorem}\label{thm:equivalence}
            Let us assume $\IR^4$-triviality.\footnote{See around~\eqref{eq:ambitwistorGauge} for the notion of $\IR^4$-triviality.} Then, the $\caN=3$ supersymmetric extension of the map $\sfE$ defined in~\eqref{eq:DGLAMorphism} is a cyclic-structure preserving weak quasi-isomorphism between the differential graded Lie algebras~$\frL_{\rm YM_1}$ and~$\frL_{\rm CR,\,tw,\,red}$ defined in~\eqref{eq:SYMDGLATwistedReduced} and~\eqref{eq:DGLAFirstOrderYM}. Put differently, the twisted CR holomorphic Chern--Simons theory defined by the Batalin--Vilkovisky CR ambitwistor action~\eqref{eq:twistedCRCSActionBV} is semi-classically equivalent to $\caN=3$ supersymmetric Yang--Mills theory defined by the Batalin--Vilkovisky space-time action~\eqref{eq:firstOrderYMActionBV}
        \end{theorem}
        
        \begin{proof}
            A lengthy but straightforward computation shows that the $\caN=3$ supersymmetric extension of the map $\sfE$ defines a weak morphism of differential graded Lie algebras in the sense of~\eqref{eq:LinftyMorphism}. Since $\sfE_1=\sfe$ descends to an isomorphism on cohomology, this shows that it is a quasi-isomorphism. Finally, the cyclicity-preserving property follows from the $\caN=3$ supersymmetric extension of \cref{eq:cyclicityPreservingProperty}.
        \end{proof}

        \noindent 
        Note that if we let $\caF_{\rm CR,\,tw,\,red}$ and $\caF_{\rm YM_1}$ be the equations of motion of the twisted CR holomorphic Chern--Simons theory and the first-order $\caN=3$ supersymmetric Yang--Mills theory in their Batalin--Vilkovisky forms, then the quasi-isomorphism relations~\eqref{eq:LinftyMorphism} amount to the fact that $\caF_{\rm CR,\,tw,\,red}\circ\sfE=\sfE\circ\caF_{\rm YM_1}$, that is, the embedding commutes with applying the equations of motion.
        
        Note also that the above theorem does not imply that $\frL_{\rm YM_1}$ is obtained from $\frL_{\rm CR,\,tw,\,red}$ by integrating out some of the fields: the latter corresponds to a homotopy transfer, which is a stronger requirement than that of a quasi-isomorphism; see e.g.~\cite{JalaliFarahani:2023sfq} for a discussion of this point. We have, however, also the following result.
        
        \begin{theorem}\label{thm:homotopy-transfer}
            There is a quasi-isomorphism from the differential graded Lie algebra $\frL_{\rm CR,\,tw,\,red}$ defined in~\eqref{eq:SYMDGLATwistedReduced} to the differential graded Lie algebra $\frL_{\rm YM_1}$ defined in~\eqref{eq:DGLAFirstOrderYM} which is computed by homotopy transfer. Put differently, integrating out the fields complementary to the image of the $\caN=3$ supersymmetric extension of the embedding $\sfe$ defined in~\eqref{eq:spaceTimeAnsatzBV} in the CR ambitwistor action~\eqref{eq:twistedCRCSActionBV} yields the $\caN=3$ supersymmetric extension of the space-time action~\eqref{eq:firstOrderYMActionBV}.
        \end{theorem}

        \begin{proof}
            Because the $\caN=3$ supersymmetric extension of $\sfe$ defined in~\eqref{eq:spaceTimeAnsatzBV} is an injective quasi-isomorphism between the cochain complexes underlying $\frL_{\rm YM_1}$ and $\frL_{\rm CR,\,tw,\,red}$, by \cref{prop:contractingHomotopyAlwaysExists}, we have a special deformation retract\footnote{i.e.~\eqref{eq:deformationRetract} with the conditions~\eqref{eq:specialDeformationRetract} satisfied}
            \begin{equation}\label{eq:anSDR}
                \begin{tikzcd}
                    \ar[loop,out=170,in=190,distance=20,"\sfh" left] (\frL_{\rm CR,\,tw,\,red},\bar\partial_{\rm CR,\,tw,\,red})\arrow[r,shift left]{}{\sfp} & (\frL_{\rm YM_1},\mu_1)\arrow[l,shift left]{}{\sfe}\,,
                \end{tikzcd}
            \end{equation}
            where $\mu_1$ is the differential of the cochain complex in~\eqref{eq:DGLAFirstOrderYM} and $\bar\partial_{\rm CR,\,tw,\,red}$ as defined in~\eqref{eq:CRantiholomorphicDifferentialTwistedReduced}. 

            To construct the projection $\sfp$ explicitly, as for $\sfe$, we restrict our discussion to the R-singlets for simplicity, that is, to the gluonic sector. The $\caN=3$ supersymmetric extension follows straightforwardly. In particular, we expand $\caA\in\Omega^{0,\bullet}_{\rm CR,\,tw,\,red}(F)\otimes\frg$ as $\caA=C+A+A^++C^+$ with
            \begin{equation}
                \begin{aligned}
                    C\ &=\ c+\cdots~,
                    \\
                    A\ &=\ \hat v^\rmF\left\{A^\rmF_{\alpha\dot\alpha}\mu^\alpha\lambda^{\dot\alpha}-\theta^i\eta_i\left(\frac{B^\rmF_{\dot\alpha\dot\beta}\lambda^{\dot\alpha}\hat\lambda^{\dot\beta}}{|\lambda|^2}-\frac{B^\rmF_{\alpha\beta}\mu^\alpha\hat\mu^\beta}{|\mu|^2}\right)\right\}
                    \\
                    &\kern1.5cm+\hat v^\rmL\left\{-\theta^i\eta_i\frac{A^\rmL_{\alpha\dot\alpha}\hat\mu^\alpha\lambda^{\dot\alpha}}{|\mu|^2}-(\theta^i\eta_i)^2\frac{3B^\rmL_{\alpha\beta}\hat\mu^\alpha\hat\mu^\beta}{2|\mu|^4}\right\}
                    \\
                    &\kern1.5cm+\hat v^\rmR\left\{\theta^i\eta_i\frac{A^\rmR_{\alpha\dot\alpha}\mu^\alpha\hat\lambda^{\dot\alpha}}{|\lambda|^2}-(\theta^i\eta_i)^2\frac{3B^\rmR_{\dot\alpha\dot\beta}\hat\lambda^{\dot\alpha}\hat\lambda^{\dot\beta}}{2|\lambda|^4}\right\}+\cdots~,
                    \\
                    A^+\ &\coloneqq\ \hat v^\rmL\wedge\hat v^\rmR\left\{(\theta^i\eta_i)^2\left(\frac{B^{+\rmL\rmR}_{\dot\alpha\dot\beta}\lambda^{\dot\alpha}\hat\lambda^{\dot\beta}}{|\lambda|^2}-\frac{B^{+\rmL\rmR}_{\alpha\beta}\mu^\alpha\hat\mu^\beta}{|\mu|^2}\right)+(\theta^i\eta_i)^3\frac{A^{+\rmL\rmR}_{\alpha\dot\alpha}\hat\mu^\alpha\hat\lambda^{\dot\alpha}}{|\mu|^2|\lambda|^2}\right\}
                    \\
                    &\kern1cm+\hat v^\rmF\wedge\hat v^\rmL\left\{-\theta^i\eta_iB^{+\rmF\rmL}_{\dot\alpha\dot\beta}\lambda^{\dot\alpha}\lambda^{\dot\beta}-(\theta^i\eta_i)^2\frac{5\hat\mu^\alpha\lambda^{\dot\alpha}}{6|\mu|^2}A^{+\rmF\rmL}_{\alpha\dot\alpha}\right\}
                    \\
                    &\kern1cm+\hat v^\rmF\wedge\hat v^\rmR\left\{\theta^i\eta_iB^{+\rmF\rmR}_{\alpha\beta}\mu^\alpha\mu^\beta-(\theta^i\eta_i)^2\frac{5\mu^\alpha\hat\lambda^{\dot\alpha}}{6|\lambda|^2}A^{+\rmF\rmR}_{\alpha\dot\alpha}\right\}+\cdots~,
                    \\
                    C^+\ &=\ \hat v^\rmF\wedge\hat v^\rmL\wedge\hat v^\rmR\,(\theta^i\eta_i)^3\frac43c^{+\rmF\rmL\rmR}+\cdots~, 
                \end{aligned}
            \end{equation}
            where all the coefficients $A^\rmF_{\alpha\dot\alpha},\ldots$ depend only on $x^{\alpha\dot\alpha}$ and the ellipses denote all the other terms that are possible in the $(\theta^i,\eta_i)$-expansion as well as in the Kaluza--Klein expansion along $\IC P^1\times\IC P^1$. The projection $\sfp$ will then project onto linear combinations of these coefficients whose explicit forms are determined by the requirement that $\sfp\circ\sfe=\sfid$. Explicitly, we may set
            \begin{subequations}\label{eq:PenroseIntegralFormulas}
                \begin{equation}
                    \resizebox{\hsize}{!}{$
                        \begin{gathered}
                            c\ \coloneqq\ \int{\rm vol}\,(\theta^i\eta_i)^3\,C~,
                            \\
                            \begin{pmatrix}
                                A_{\alpha\dot\alpha}
                                \\
                                B_{\dot\alpha\dot\beta}
                                \\
                                B_{\alpha\beta}
                            \end{pmatrix}
                            \ \coloneqq\ \int{\rm vol}
                            \begin{pmatrix}
                                (\theta^i\eta_i)^2\Big(\frac{A_\rmL\mu_\alpha\hat\lambda_{\dot\alpha}}{|\lambda|^2}-\frac{A_\rmR\hat\mu_\alpha\lambda_{\dot\alpha}}{|\mu|^2}\Big)+(\theta^i\eta_i)^3\Big(\frac{A_\rmF\hat\mu_\alpha\hat\lambda_{\dot\alpha}}{|\mu|^2|\lambda|^2}+V_\rmF\sfG(A)\mu_\alpha\lambda_{\dot\alpha}\Big)
                                \\
                                \theta^i\eta_i\sfG(\hat V_\rmR\intprod\bar\partial_{\rm CR,\,tw,\,red}A)\lambda_{\dot\alpha}\lambda_{\dot\beta}-(\theta^i\eta_i)^2\frac{\sfG(\hat V_\rmF\intprod\bar\partial_{\rm CR,\,tw,\,red}A))\lambda_{(\dot\alpha}\hat\lambda_{\dot\beta)}}{|\lambda|^2}+\theta^i\eta_i(1+\theta^i\eta_iV_\rmF)\sfG(\hat V_\rmR\intprod\bar\partial_{\rm CR,\,tw,\,red}A)\lambda_{\dot\alpha}\lambda_{\dot\beta}-(\theta^i\eta_i)^3\frac{V_\rmF\sfG(\hat V_\rmF\intprod\bar\partial_{\rm CR,\,tw,\,red}A)\lambda_{(\dot\alpha}\hat\lambda_{\dot\beta)}}{|\lambda|^2}
                                \\
                                \theta^i\eta_i\sfG(\hat V_\rmL\intprod\bar\partial_{\rm CR,\,tw,\,red}A)\mu_\alpha\mu_\beta+(\theta^i\eta_i)^2\frac{\sfG(\hat V_\rmF\intprod\bar\partial_{\rm CR,\,tw,\,red}A)\mu_{(\alpha}\hat\mu_{\beta)}}{|\mu|^2}+\theta^i\eta_i(-1+\theta^i\eta_iV_\rmF)\sfG(\hat V_\rmL\intprod\bar\partial_{\rm CR,\,tw,\,red}A)\mu_\alpha\mu_\beta+(\theta^i\eta_i)^3\frac{V_\rmF\sfG(\hat V_\rmF\intprod\bar\partial_{\rm CR,\,tw,\,red}A)\mu_{(\alpha}\hat\mu_{\beta)}}{|\mu|^2}
                            \end{pmatrix},
                            \\
                            \begin{pmatrix}
                                A^+_{\alpha\dot\alpha}
                                \\
                                B^+_{\dot\alpha\dot\beta}
                                \\
                                B^+_{\alpha\beta}
                            \end{pmatrix}
                            \ \coloneqq\ \int{\rm vol}
                            \begin{pmatrix}
                                \mu_\alpha\lambda_{\dot\alpha}A^+_{\rmL\rmR}+\theta^i\eta_i\Big(\frac{A^+_{\rmF\rmL}\mu_{\alpha}\hat\lambda_{\dot\alpha}}{|\lambda|^2}+\frac{A^+_{\rmF\rmR}\hat\mu_{\alpha}\lambda_{\dot\alpha}}{|\mu|^2}\Big)+\theta^i\eta_i\mu_\alpha\lambda_{\dot\alpha}V_\rmF A^+_{\rmL\rmR}+\theta^i\eta_i\Big\{\frac{(-1+\theta^i\eta_iV_\rmF)A^+_{\rmF\rmL}\mu_{\alpha}\hat\lambda_{\dot\alpha}}{|\lambda|^2}+\frac{(1+\theta^i\eta_iV_\rmF)A^+_{\rmF\rmR}\hat\mu_{\alpha}\lambda_{\dot\alpha}}{|\mu|^2}\Big\}
                                \\
                                \theta^i\eta_i\sfG(\hat V_\rmR\intprod A^+)\lambda_{\dot\alpha}\lambda_{\dot\beta}-(\theta^i\eta_i)^2\frac{\sfG(\hat V_\rmF\intprod A^+)\lambda_{(\dot\alpha}\hat\lambda_{\dot\beta)}}{|\lambda|^2}-\theta^i\eta_i\frac{A^+_{\rmL\rmR}\lambda_{(\dot\alpha}\hat\lambda_{\dot\beta)}}{|\lambda|^2}-(\theta^i\eta_i)^2\frac{A^+_{\rmF\rmL}\hat\lambda_{\dot\alpha}\hat\lambda_{\dot\beta}}{|\lambda|^4}
                                \\
                                \theta^i\eta_i\sfG(\hat V_\rmL\intprod A^+)\mu_\alpha\mu_\beta+(\theta^i\eta_i)^2\frac{\sfG(\hat V_\rmF\intprod A^+)\mu_{(\alpha}\hat\mu_{\beta)}}{|\mu|^2} -\theta^i\eta_i\frac{A^+_{\rmL\rmR}\mu_{(\alpha}\hat\mu_{\beta)}}{|\mu|^2}-(\theta^i\eta_i)^2\frac{A^+_{\rmF\rmR}\hat\mu_\alpha\hat\lambda_\beta}{|\mu|^4}
                            \end{pmatrix},
                            \\
                            c^+\ \coloneqq\ \int{\rm vol}\,(1+\theta^i\eta_iV_\rmF)C^+_{\rmF\rmL\rmR}
                        \end{gathered}
                    $}
                \end{equation}
                up to some irrelevant overall multiplicative constants and where\footnote{See also~\eqref{eq:reducedCRExpressions}.} 
                \begin{equation}
                    {\rm vol}\ \coloneqq\ v^\rmL\wedge\hat v^\rmL\wedge v^\rmR\wedge\hat v^\rmR\otimes\rmd^3\eta\,\rmd^3\theta
                \end{equation}
                as well as $A_\rmF\coloneqq\hat V_\rmF\intprod A$, $A_\rmL\coloneqq\hat V_\rmL\intprod A$, $A_\rmR\coloneqq\hat V_\rmR\intprod A$, $A^+_{\rmL\rmR}\coloneqq\hat V_\rmR\intprod V_\rmL\intprod A^+$, etc.~and with the basis vector fields and basis one-forms are as defined in~\eqref{eq:tangentBundleDistributionsCRSuperAmbitwistorSpaceTwisted},~\eqref{eq:cotangentBundleDistributionCRSuperAmbitwistorSpaceTwisted}, and~\eqref{eq:holomorphicCotangentBundleDistributionCRSuperAmbitwistorSpaceTwisted}. In addition, we may take $\sfG:\Omega_{\rm CR,\,tw,\,red}^{0,1}(F)\rightarrow\Omega_{\rm CR,\,tw,\,red}^{0,0}(F)$ to be Green operator for $\bar\partial_{\rm CR,\,tw,\,red}$ which would mean, however, fixing a metric on $F$.\footnote{Note that the appearance of $\sfG$ is not surprising as this follows from the general considerations; see the proof of \cref{prop:contractingHomotopyAlwaysExists}.} To avoid this, here we instead take it to be~\cite{Borsten:2022vtg}
                \begin{equation}
                    \sfG\ \coloneqq\ \frac{\sfb}{\BBox} 
                \end{equation}
                with
                \begin{equation}
                    \begin{gathered}
                        \sfb\ \coloneqq\ 8\hat V_\rmF\intprod\caL_{V_\rmF}\ =\ 8V_\rmF\hat V_\rmF\intprod~,\\
                        \BBox\ \coloneqq\ 8V_\rmF\hat V_\rmF\ =\ \wave+4(V_\rmF\hat V_\rmF-E_\rmW E_{\hat\rmW})~,~~~
                        \wave\ \coloneqq\ 2\partial^{\alpha\dot\alpha}\partial_{\alpha\dot\alpha}~.
                    \end{gathered}
                \end{equation}
            \end{subequations}
            It is not too difficult to see that $\sfp\circ\sfe=\sfid$ using the explicit form~\eqref{eq:spaceTimeAnsatzBV} of $\sfe$ up to some irrelevant overall multiplicative constants. It also follows that $\sfp$ is a cochain map by a direct calculation similar to what is done to prove \cref{prop:injectiveCochainMorphism} together with the identity $\sfG(\hat V_{\rmF,\rmL,\rmR}\intprod\bar\partial_{\rm CR,\,tw,\,red}A)=\hat V_{\rmF,\rmL,\rmR}\sfG(A)-\hat V_{\rmF,\rmL,\rmR}\intprod A$ which follows from $\big[\hat V_{\rmF,\rmL,\rmR},\frac{V_\rmF}{\BBox}\big]=0$. 

            Importantly, using the explicit expressions of $\sfE_2$ and $\sfE_3$ from~\eqref{eq:DGLAMorphism} and since $\sfG$ does not alter the dependence on the CR holomorphic fermionic coordinates, it now immediately follows that
            \begin{equation}\label{eq:higherMorphismProjections}
                \sfp\circ\sfE_2\ =\ 0\ =\ \sfp\circ\sfE_3~.
            \end{equation} 
            Furthermore, for $i=2$, the relation~\eqref{eq:homotopTransferHigherProducts} states that\footnote{The signs given as $\pm$ are irrelevant to our discussion.}
            \begin{equation}
                \begin{aligned}
                    &\sfe(\mu_2(\ell_1,\ell_2))\pm\sfE_2(\mu_1(\ell_1),\ell_2)\pm\sfE_2(\mu_1(\ell_2),\ell_1)
                    \\
                    &\kern1cm=\ \pm\bar\partial_{\rm CR,\,tw,\,red}(\sfE_2(\ell_1,\ell_2))+[\sfe(\ell_1),\sfe(\ell_2)]
                \end{aligned}
            \end{equation}
            for all $\ell_1,\ell_2\in\frL_{\rm YM_1}$ and with $\mu_2$ as given in~\eqref{eq:DGLAFirstOrderYM}. Upon applying $\sfp$ to this equation and using~\eqref{eq:higherMorphismProjections} as well as  the facts that $\sfp$ is a cochain map and that $\sfp\circ\sfe=\sfid$, we obtain
            \begin{equation}
                \mu_2(\ell_1,\ell_2)\ =\ \sfp([\sfe(\ell_1),\sfe(\ell_2)])~.
            \end{equation}
            This, however, is precisely the binary product~\eqref{eq:homotopTransferHigherProducts} obtained from homotopy transfer via the special deformation retract~\eqref{eq:anSDR}. 

            It remains to show that all higher products in~\eqref{eq:homotopTransferHigherProducts} vanish. This can be done by a lengthy direct computation, constructing an explicit $\sfh$ along the lines of~\cref{prop:contractingHomotopyAlwaysExists}. Alternatively, we can simply argue that our formalism preserves all space-time and gauge symmetries, and there are simply not quartic or higher interaction vertices that can be consistently constructed in $\frL_{\rm YM_1}$ from the field content which respect translation and conformal symmetry. See also~\cref{rmk:alternativeProof} for an alternative (and much shorter) proof that makes use of the uniqueness of maximally supersymmetric Yang--Mills theory.
        \end{proof}

        \begin{remark}\label{rmk:alternativeProof}
            Note that the last argument in the above proof can actually be extended, so that after the existence of the deformation retract~\eqref{eq:anSDR} is established, \cref{thm:homotopy-transfer} follows automatically. Indeed, the theory $\frL_\sfT$ obtained on the graded vector space $\frL_{\rm YM_1}$ by homotopy transfer has the same field content and kinematic terms as $\frL_{\rm YM_1}$. The quasi-isomorphism from \cref{thm:equivalence} implies that the tree-level scattering amplitudes of the theory $\frL_\sfT$ agrees with those of $\caN=3$ supersymmetric Yang--Mills theory; hence, the theory is not free. The setup and homotopy transfer preserve gauge and space-time symmetries. Moreover, the formalism is symmetric under the CPT-like symmetry exchanging undotted and dotted spinor indices, increasing the supersymmetry from $\caN=3$ to $\caN=4$. The superconformal formulation of an interacting $\caN=4$ supersymmetric Yang--Mills theory with the first-order field content, however is unique. Consequently, $\frL_\sfT$ has to agree with $\frL_{\rm YM_1}$.
        \end{remark}

        \begin{remark}
            The integral formulas~\eqref{eq:PenroseIntegralFormulas} can be understood as the `Dolbeault analogue' of Penrose's contour integral formulas in this CR ambitwistor setting. See e.g.~\cite{JSTOR:2990349} for analogous integral formulas in the `\v Cech formulation' for the second cohomology in the purely bosonic complex ambitwistor setting.
        \end{remark}
        
        \appendix
        \addappheadtotoc
        \appendixpage 
        
        \appendices
        
        \section{Homological considerations}
        
        \subsection{Quasi-isomorphic cochain complexes}
        
        \paragraph{Split complex supermanifolds.}
        Let $M$ be a complex manifold and $\pi:E\rightarrow M$ a holomorphic vector bundle. Using the definition of the vertical tangent bundle $V$ together with universality of the pullback, we obtain the short exact sequence
        \begin{equation}
            \begin{tikzcd}
                0\ar[r] & V\ar[r] & TE\ar[r] & \pi^*TM\ar[r] & 0~.
            \end{tikzcd}
        \end{equation}
        Its restriction $V|_M$ to $M$ can be identified with $E$ and so,
        \begin{equation}
            \begin{tikzcd}
                0\ar[r] & E\ar[r] & TE|_M\ar[r] & TM\ar[r] & 0~.
            \end{tikzcd}
        \end{equation}
        This sequence splits canonically by the zero section $M\rightarrow E$. Hence, we have the canonical identification
        \begin{equation}
            TE|_M\ \cong\ TM\oplus E~.
        \end{equation}
        Since $\pi$ is a homotopy equivalence, it follows that the unrestricted sequence also splits, though non-canonically, that is,
        \begin{equation}\label{eq:nonCanonicalSplitting}
            TE\ \cong\ \pi^*TM\oplus V~.
        \end{equation}
        Hence,
        \begin{equation}\label{eq:splitting}
            \mbox{$\bigwedge$}^{0,p}\,T^*E\ \cong\ \bigoplus_{r+s=p}\mbox{$\bigwedge$}^{0,r}\pi^*T^*M\otimes\mbox{$\bigwedge$}^{0,s}V^*~.
        \end{equation}
        
        Let $\bar\partial$ be the anti-holomorphic exterior derivative on $E[1]$, where $[1]$ denotes the (Gra{\ss}mann-)degree shift by one of the fibres of $E$, and let $\bar\partial_{\rm red}$ be the anti-holomorphic exterior derivative on $M$, respectively. The manifold $E[1]$ is what is known as a globally \uline{split complex supermanifold}.\footnote{Note that in the smooth category, every supermanifold is (non-canonically) globally split due to~\cite{JSTOR:1998201}. This is no longer the case in the complex category essentially because there is no holomorphic partition of unity.}
        
        \paragraph{Quasi-isomorphic cochain complexes.}
        For any split complex supermanifold $E[1]$, we now have the following result. It can be proved abstractly using spectral sequences (see e.g.~\cite{Movshev:2004ub}) but here we provide an elementary constructive proof as we shall need its extension to the differential graded Lie algebra in \cref{app:homologicalPerturbations}.
        
        \begin{proposition}\label{prop:quasiIsoSplitSupermanifolds} 
            The cochain complexes $(\Omega^{0,\bullet}(E[1]),\bar\partial)$ and $(\Omega^{0,\bullet}(M,\bigwedge^\bullet E^*),\bar\partial_{\rm red})$ are quasi-isomorphic with the quasi-isomorphism induced by coboundary transformations.
        \end{proposition}
        
        \begin{proof}
            Firstly,~\eqref{eq:splitting} implies
            \begin{equation}
                \Omega^{0,p}(E[1])\ \cong\ \bigoplus_{r+s=p}\Gamma\big(E[1],\mbox{$\bigwedge$}^{0,r}\pi^*T^*M\otimes\mbox{$\bigwedge$}^{0,s}(V[1])^*\big)~.
            \end{equation}
            Furthermore, the smooth functions on $E[1]$ which are holomorphic in the fermionic coordinates can be identified with $\Gamma(M,\bigwedge^\bullet E^*)$ and so, $\Gamma(M,\bigwedge^\bullet E^*)\hookrightarrow\scC^\infty(E[1])=\Omega^{0,0}(E[1])$. Thus, we have the inclusions
            \begin{equation}\label{eq:inclusions}
                \begin{tikzcd}
                    \Omega^{0,p}\big(M,\mbox{$\bigwedge$}^\bullet E^*\big)\arrow[hookrightarrow]{r}{} & \Gamma\big(E[1],\mbox{$\bigwedge$}^{0,p}\pi^*T^*M\big)\arrow[hookrightarrow]{r}{} & \Omega^{0,p}(E[1])
                \end{tikzcd}
            \end{equation}
            for all $p$ and, consequently, the commutative diagrams
            \begin{equation}
                \begin{tikzcd}
                    \Omega^{0,0}(E[1])\ar[r,"\bar\partial"] & \Omega^{0,1}(E[1])\ar[r,"\bar\partial"] & \Omega^{0,2}(E[1])\ar[r,"\bar\partial"] & \cdots
                    \\
                    \Omega^{0,0}(M,\bigwedge^\bullet E^*)\ar[r,"\bar\partial_{\rm red}"]\arrow[hookrightarrow]{u}{\iota} & \Omega^{0,1}(M,\bigwedge^\bullet E^*)\ar[r,"\bar\partial_{\rm red}"]\arrow[hookrightarrow]{u}{\iota} & \Omega^{0,2}(M,\bigwedge^\bullet E^*)\ar[r,"\bar\partial_{\rm red}"]\arrow[hookrightarrow]{u}{\iota} & \cdots
                \end{tikzcd}
            \end{equation}
            where the $\iota$ are given by the compositions of the inclusions in~\eqref{eq:inclusions}. In addition, the existence of the embedding also implies that
            \begin{equation}
                H^0(E[1])\ \cong\ H^0\big(M,\mbox{$\bigwedge$}^\bullet E^*\big)
            \end{equation}
            establishing the claim for $p=0$.
            
            Next, let us analyse the underlying cohomology for $p>0$. To this end, let $(z^a,\bar z^{\bar a},\eta^\alpha,\bar\eta^{\bar\alpha})$ be local coordinates on $E[1]$ with $(z^a,\bar z^{\bar a})$ bosonic base coordinates and $(\eta^\alpha,\bar\eta^{\bar\alpha})$ fermionic fibre coordinates. Then, the anti-holomorphic multivector fields $\frX^{0,\bullet}(E[1])$ on $E[1]$ are generated by
            \begin{subequations}\label{eq:basisLift}
                \begin{equation}\label{eq:basisLiftVectorFields}
                    \bar E_{\bar a}\ \coloneqq\ \parder{\bar z^{\bar a}}+\parder{\bar z^{\bar a}}\intprod\bar\eta^{\bar\beta}\Gamma_{\bar\beta}{}^{\bar\alpha}\parder{\bar\eta^{\bar\alpha}}
                    \eand
                    \bar E_{\bar\alpha}\ \coloneqq\ \parder{\bar\eta^{\bar\alpha}}~,
                \end{equation}
                where $\bar E_{\bar a}$ is the horizontal lift of $\parder{\bar z^{\bar a}}$ to $E[1]$ induced by the splitting~\eqref{eq:nonCanonicalSplitting} with the associated connection one-form $\Gamma_{\bar\beta}{}^{\bar\alpha}\coloneqq\rmd\bar z^{\bar a}\Gamma_{\bar a\bar\beta}{}^{\bar\alpha}$. Dually, we have
                \begin{equation}\label{eq:basisLiftDifferentialForms}
                    \bar e^{\bar a}\ \coloneqq\ \rmd\bar z^{\bar a}
                    \eand
                    \bar e^{\bar\alpha}\ \coloneqq\ \rmd\bar\eta^{\bar\alpha}-\bar\eta^{\bar\beta}\Gamma_{\bar\beta}{}^{\bar\alpha}~,
                \end{equation}
            \end{subequations}
            and which generate the anti-holomorphic differential forms $\Omega^{0,\bullet}(E[1])$ on $E[1]$. Evidently,
            \begin{subequations}
                \begin{equation}
                    \begin{aligned}
                        [\bar E_{\bar a},\bar E_{\bar b}]\ =\ \bar\eta^{\bar\alpha}R_{\bar a\bar b\bar\alpha}{}^{\bar\beta}\bar E_{\bar\beta}~,
                        \quad
                        [\bar E_{\bar a},\bar E_{\bar\alpha}]\ =\ -\Gamma_{\bar a\bar\alpha}{}^{\bar\beta}\bar E_{\bar\beta}~,
                        \eand
                        [\bar E_{\bar\alpha},\bar E_{\bar\beta}]\ =\ 0~,
                    \end{aligned}
                \end{equation}
                where $R_{\bar a\bar b\bar\alpha}{}^{\bar\beta}$ is the curvature of $\Gamma_{\bar a\bar\beta}{}^{\bar\alpha}$ and so,
                \begin{equation}\label{eq:antiHolomorphicMaurerCartan}
                    \bar\partial\bar e^{\bar a}\ =\ 0
                    \eand
                    \bar\partial\bar e^{\bar\alpha}\ =\ \tfrac12\bar e^{\bar b}\wedge\bar e^{\bar a}\,\bar\eta^{\bar\beta}R_{\bar a\bar b\bar\beta}{}^{\bar\alpha}+\bar e^{\bar\beta}\wedge\bar e^{\bar a}\Gamma_{\bar a\bar\beta}{}^{\bar\alpha}~.
                \end{equation}
            \end{subequations}
            Using the basis $(0,1)$-forms~\eqref{eq:basisLiftDifferentialForms}, an element $\omega\in\Omega^{0,p}(E[1])$ looks like
            \begin{equation}
                \omega\ =\ \sum_{r+s=p}\tfrac{1}{r!s!}\,\bar e^{\bar a_1}\wedge\ldots\wedge\bar e^{\bar a_r}\otimes\bar e^{\bar\alpha_1}\wedge\ldots\wedge\bar e^{\bar\alpha_s}\,\omega_{\bar\alpha_s\cdots\bar\alpha_1\bar a_r\cdots\bar a_1}(z,\bar z,\eta,\bar\eta)~.
            \end{equation}
            Suppose now that $\omega$ is a representative of an element $[\omega]$ of the cohomology group $H^p(E[1])$. 
            Then, using~\eqref{eq:antiHolomorphicMaurerCartan}, it is not too difficult to see that $\bar\partial\omega=0$ yields
            \begin{equation}\label{eq:holoPurelyFerm}
                \bar E_{(\bar\alpha_1}\omega_{\bar\alpha_2\cdots\bar\alpha_{p+1})}\ =\ 0
            \end{equation}
            for the purely fermionic component. Next, consider $\tilde c^{(1)}\in\Omega^{0,p-1}(E[1])$ given by
            \begin{equation}
                \tilde c^{(1)}\ \coloneqq\ \tfrac{1}{(p-1)!}\,\bar\eta^{\bar\alpha_1}\bar e^{\bar\alpha_2}\wedge\ldots\wedge\bar e^{\bar\alpha_p}\,\omega_{\bar\alpha_p\cdots\bar\alpha_1}~.
            \end{equation}
            This is invariant under bundle isomorphisms and hence globally defined. A short computation then shows that the purely fermionic component of $\bar\partial\tilde c^{(1)}$ is 
            \begin{equation}\label{eq:actionEbarCTilde}
                \bar E_{(\bar\alpha_1}\tilde c^{(1)}_{\bar\alpha_2\cdots\bar\alpha_p)}\ =\ \big(1+\tfrac1p\Upsilon)\omega_{\bar\alpha_1\cdots\bar\alpha_p}~,
            \end{equation}
            where we have used~\eqref{eq:holoPurelyFerm} and introduced the globally defined anti-holomorphic Euler vector field $\Upsilon\coloneqq\bar\eta^{\bar\alpha}\bar E_{\bar\alpha}$. Since the components of $\omega$ are polynomials of the fermionic coordinates, the action of $\Upsilon$ will return only non-negative integers and so, the inverse of $1+\tfrac1p\Upsilon$ is well defined for all $p=1,2,3,\ldots$ with
            \begin{equation}
                \frac{1}{1+\frac1p\Upsilon}\ =\ \int_0^\infty\rmd t\,\rme^{-t(1+\frac1p\Upsilon)}~.
            \end{equation}
            Furthermore, since $[\bar E_{\bar\alpha},\Upsilon]=\bar E_{\bar\alpha}$, we have
            \begin{equation}
                \begin{aligned}
                    \frac{1}{1+\frac1p\Upsilon}\,\bar E_{\bar\alpha}\ &=\ \int_0^\infty\rmd t\,\rme^{-t(1+\frac1p)\Upsilon}\,\bar E_{\bar\alpha}
                    \\
                    &=\ \int_0^\infty\rmd t\,\rme^{-t(1+\frac1p\Upsilon)}\,\bar E_{\bar\alpha}\,\rme^{t(1+\frac1p\Upsilon)}\,\rme^{-t(1+\frac1p\Upsilon)}
                    \\
                    &=\ \int_0^\infty\rmd t\,\underbrace{\rme^{-\frac tp\Upsilon}\,\bar E_{\bar\alpha}\,\rme^{\frac tp\Upsilon}}_{=\,\rme^{\frac tp}\bar E_{\bar\alpha}}\,\rme^{-t(1+\frac1p\Upsilon)}
                    \\
                    &=\ \bar E_{\bar\alpha}\underbrace{\int_0^\infty\rmd t\,\rme^{-t(1-\frac1p+\frac1p\Upsilon)}}_{\eqqcolon\,D(\Upsilon)}
                \end{aligned}
            \end{equation}
            when understood as acting on expressions that are at least of linear order in the anti-holomorphic fermionic coordinates; in this case, the action of $\Upsilon$ in $D(\Upsilon)$ will return only positive integers so that the derivation of $D(\Upsilon)$ and $D(\Upsilon)$ itself are also well defined for $p=1$. Consequently,~\eqref{eq:actionEbarCTilde} becomes
            \begin{equation}
                \bar E_{(\bar\alpha_1}\big(D(\Upsilon)\tilde c^{(1)}_{\bar\alpha_2\cdots\bar\alpha_p)}\big)\ =\ \omega_{\bar\alpha_1\cdots\bar\alpha_p}~.
            \end{equation}
            This shows that the component $\omega_{\bar\alpha_1\cdots\bar\alpha_p}$ is a coboundary parametrised by $c^{(1)}\in\Omega^{0,p-1}(E[1])$ with 
            \begin{equation}
                c^{(1)}\ \coloneqq\ \tfrac{1}{(p-1)!}\bar e^{\bar\alpha_1}\wedge\ldots\wedge\bar e^{\bar\alpha_{p-1}}\,D(\Upsilon)\tilde c^{(1)}_{\bar\alpha_{p-1}\cdots\bar\alpha_1}~.
            \end{equation}
            Note that this is globally defined. Evidently, we have not used the full freedom of coboundary transformations yet. Let $\omega^{(1)}\coloneqq\omega-\bar\partial c^{(1)}$. Since $\omega^{(1)}_{\bar\alpha_p\cdots\bar\alpha_1}$ is absent, $\bar\partial\omega^{(1)}=0$ implies that $\omega^{(1)}_{\bar\alpha_{p-1}\cdots\bar\alpha_1\bar a}$ is a coboundary by similar arguments we have just given, and so we can transform it away by using a coboundary parametrised by $c^{(2)}\in\Omega^{0,p-1}(E[1])$ of the form
            \begin{equation}
                c^{(2)}\ \coloneqq\ \tfrac{1}{(p-2)!}\,\bar e^{\bar a}\otimes\bar e^{\bar\alpha_1}\wedge\ldots\wedge\bar e^{\bar\alpha_{p-2}}\,c^{(2)}_{\bar\alpha_{p-2}\cdots\bar\alpha_1\bar a}~.
            \end{equation}
            We can now iterate this procedure to eventually arrive at a $(0,p)$-form
            \begin{equation}
                \omega^{(p)}\ \coloneqq\ \tfrac{1}{p!}\,\bar e^{\bar a_1}\wedge\ldots\wedge\bar e^{\bar a_p}\,\omega^{(p)}_{\bar a_p\cdots\bar a_1}(z,\bar z,\eta)\ =\ \tfrac{1}{p!}\,\rmd\bar z^{\bar a_1}\wedge\ldots\wedge\rmd\bar z^{\bar a_p}\,\omega^{(p)}_{\bar a_p\cdots\bar a_1}(z,\bar z,\eta)
            \end{equation}
            whose components depend holomorphically\footnote{Once all fermionic directions have been transformed away, the condition $\bar\partial\omega^{(p)}=0$ implies that the remaining components must be holomorphic in the fermionic coordinates.} on $\eta^\alpha$ and which also belongs to the equivalence class $[\omega]$. Since, $0=\bar\partial\omega^{(p)}=\bar\partial_{\rm red}\omega^{(p)}$ and since, as mentioned before, the smooth functions on $E[1]$ which are holomorphic in $\eta^\alpha$ can be identified with $\Gamma(M,\bigwedge^\bullet E^*)$, the $(0,p)$-form $\omega^{(p)}$ is also a representative of an equivalence class in $H^p(M,\bigwedge^\bullet E^*)$. 
            
            Altogether, we have established the desired isomorphisms
            \begin{equation}
                H^\bullet(E[1])\ \cong\ H^\bullet\big(M,\mbox{$\bigwedge$}^\bullet E^*\big)\,.
            \end{equation}
        \end{proof}
        
        \subsection{Homological perturbations}\label{app:homologicalPerturbations}
        
        \paragraph{Deformation retracts.}
        One method of obtaining a semi-classically equivalent field theory from a perturbative action is to integrate out parts of the field content. From the homotopical algebra perspective on field theory, in which perturbative field theories with action principle are regarded as cyclic $L_\infty$-algebras, this is done by homotopy transfer, which is based on the homological perturbation lemma~\cite{gugenheim1991perturbation,Crainic:0403266}. More generally, the well known tree-level perturbative expansion in terms of Feynman diagrams is mathematically captured by homological perturbation theory. 
        
        The starting point for homological perturbation theory is a \uline{deformation retract}, that is, a diagram 
        \begin{subequations}\label{eq:deformationRetract}
            \begin{equation}
                \begin{tikzcd}
                    \ar[loop,out=160,in=200,distance=20,"\sfh" left]\sfC^{(1)}\arrow[r,shift left]{}{\sfp} & \sfC^{(2)}\arrow[l,shift left]{}{\sfe}~,
                \end{tikzcd}
            \end{equation}
            of cochain complexes $\sfC^{(1)}$ and $\sfC^{(2)}$ with differentials $\sfd^{(1)}$ and $\sfd^{(2)}$, respectively, $\sfp$ and $\sfe$ are morphisms of cochain complexes, and $\sfh$ is a \uline{contracting homotopy}, that is, a morphism of graded vector spaces which has degree $-1$, and these maps obey\footnote{Note that these two conditions promote $\sfp$ and $\sfe$ to quasi-isomorphisms of cochain complexes.}
            \begin{equation}
                \sfid-\,\sfe\circ\sfp\ =\ \sfd^{(1)}\circ\sfh+\sfh\circ\sfd^{(1)}\eand
                \sfp\circ\sfe\ =\ \sfid~.
            \end{equation}
        \end{subequations}
        A deformation retract can always be turned into a \uline{special deformation retract}~\cite{Crainic:0403266,Loday:2012aa} by a redefinition of $\sfh$ such that the maps also satisfy 
        \begin{equation}\label{eq:specialDeformationRetract}
            \sfp\circ\sfh\ =\ 0~,
            \quad
            \sfh\circ\sfe\ =\ 0~,
            \eand
            \sfh\circ\sfh\ =\ 0~.
        \end{equation}
        Explicitly, a redefinition that does this job is given by~\cite{Crainic:0403266,Loday:2012aa} 
        \begin{equation}\label{eq:hRefinition}
            \sfh\ \to\ (\id-\,\sfe\circ\sfp)\circ\sfh\circ(\id-\,\sfe\circ\sfp)\circ\sfd^{(1)}\circ(\id-\,\sfe\circ\sfp)\circ\sfh\circ(\id-\,\sfe\circ\sfp)~.
        \end{equation}
        
        The following describes the change of the map $\sfe$ in a deformation retract to the cohomology, which physically captures in particular changes of gauge.

        \begin{proposition}\label{prop:gauge_change}
            Consider a deformation retract
            \begin{equation}\label{eq:def_ret_1}
                \begin{tikzcd}
                    \ar[loop,out=160,in=200,distance=20,"\sfh" left]\sfC\arrow[r,shift left]{}{\sfp} & H^\bullet(\sfC)\arrow[l,shift left]{}{\sfe}~,
                \end{tikzcd}
            \end{equation}
            and a quasi-isomorphism $\tilde\sfe:H^\bullet(\sfC)\rightarrow\sfC$. Then,
            we have
            \begin{equation}
                \tilde\sfe\ =\ \sfe\circ\phi+\psi
            \end{equation} 
            for $\phi$ some automorphism on $H^\bullet(\sfC)$, a morphism $\psi:H^\bullet(\sfC)\rightarrow C^\bullet(\sfC)$, and
            \begin{equation}
                \begin{tikzcd}
                    \ar[loop,out=160,in=200,distance=20,"\tilde \sfh" left]\sfC\arrow[r,shift left]{}{\tilde \sfp} & H^\bullet(\sfC)\arrow[l,shift left]{}{\tilde \sfe}~,
                \end{tikzcd}
            \end{equation}
            with 
            \begin{equation}
				\begin{aligned}
                    \tilde \sfp\ \coloneqq\ \phi^{-1}\circ\sfp
                    \eand 
                    \tilde \sfh\ \coloneqq\ \sfh-\sfh\circ\psi\circ\phi^{-1}\circ\sfp
				\end{aligned}
			\end{equation}
            is also a deformation retract.
        \end{proposition}

        \begin{proof}
            The data in the deformation retract~\eqref{eq:def_ret_1} induce the \uline{Hodge--Kodaira decomposition}
            \begin{equation}
                \sfC^{p}\ \cong\ H^p(\sfC^{})\oplus\underbrace{\im(\sfh^{}_{p+1}\circ\sfd^{}_p)}_{\eqqcolon\,B^p(\sfC^{})}\oplus\underbrace{\im(\sfd^{}_{p-1}\circ\sfh^{}_p)}_{\eqqcolon\,C^p(\sfC^{})}~,
            \end{equation} 
            where we note that $\sfh\circ\sfd\circ\sfh=\sfh$ and $\sfd\circ\sfh\circ\sfd=\sfd$ so that
            \begin{equation}
                \Pi_{B^\bullet(\sfC)}\ \coloneqq\ \sfh\circ\sfd
                \eand 
                \Pi_{C^\bullet(\sfC)}\ \coloneqq\ \sfd\circ\sfh
            \end{equation} 
            are indeed projectors onto $B^\bullet(\sfC)$ and $C^\bullet(\sfC)$. Because $\tilde \sfe$ is a cochain map, its image is contained in $H^\bullet(\sfC)\oplus C^\bullet(\sfC)$. Because it is a quasi-isomorphism and hence descends to an isomorphism on cohomology, the restriction\footnote{For a linear map $f:A\oplus B\rightarrow C\oplus D$, we denote its restriction to $A\rightarrow C$ by $f|_A^C$.} of the image of $\tilde\sfe$ to $H^\bullet(\sfC)$, $\tilde\sfe|^{H^\bullet(\sfC)}$, is an isomorphism. Together, these facts imply that $\tilde\sfe=\sfe\circ\phi+\psi$.
            
            It remains to check the properties of a deformation retract. Since $\sfp|_{C^\bullet(\sfC)}=0$, we have $\tilde \sfp\circ \tilde \sfe=\sfid$. Moreover, 
            \begin{equation}
				\begin{aligned}
                    \sfid-\,\tilde\sfe\circ\tilde \sfp\ &=\ \sfid-\,\sfe\circ\sfp-\psi\circ\phi^{-1}\circ\sfp
                    \\
                    &=\ \sfd\circ \sfh+\sfh\circ\sfd-\Pi_{C^\bullet(\sfC)}\circ\psi\circ\phi^{-1}\circ\sfp
                    \\
                    &=\ \sfd\circ\sfh+\sfh\circ\sfd-\sfd\circ\sfh\circ\psi\circ\phi^{-1}\circ\sfp-\sfh\circ\psi\circ\phi^{-1}\circ\sfp\circ\sfd
                    \\
                    &=\ \sfd\circ\tilde\sfh+\tilde\sfh\circ\sfd~,
				\end{aligned}
			\end{equation}
            where we have used that $\sfp\circ\sfd=0$.
        \end{proof}

        \begin{proposition}\label{prop:contractingHomotopyAlwaysExists}
            Given an injective quasi-isomorphism $\sfe:\sfC^{(2)}\rightarrow\sfC^{(1)}$ of split cochain complexes with $\sfC^{(i)p}$ trivial for $i=1,2$ and for all $p<p_0$ and $p>p_1$ for some fixed $p_{0,1}\in\IZ$ with $p_0<p_1$, then there always exists a special deformation retract
        	\begin{equation}
            	\begin{tikzcd}
                    \ar[loop,out=160,in=200,distance=20,"\sfh" left]\sfC^{(1)}\arrow[r,shift left]{}{\sfp} & \sfC^{(2)}\arrow[l,shift left]{}{\sfe}~.
                \end{tikzcd}
            \end{equation}
        \end{proposition}

        \begin{proof}
            We only need to demonstrate that there is a deformation retract by virtue of our discussion after~\eqref{eq:deformationRetract}. We first construct a projection $\sfp$ iteratively. Since the cochain complexes are split, we have for $i=1,2$ the deformation retracts\footnote{See~e.g.~\cite[Chapter 1.4]{Weibel:1994aa} or~\cite[Appendix B]{Jurco:2018sby} for details.}
            \begin{equation}\label{eq:someSDRs}
                \begin{tikzcd}
                    \ar[loop,out=160,in=200,distance=20,"\sfh^{(i)}" left]\sfC^{(i)}\arrow[r,shift left]{}{\sfp^{(i)}} & H^\bullet(\sfC^{(i)})\arrow[l,shift left]{}{\sfe^{(i)}}
                \end{tikzcd}
            \end{equation}
            onto the cohomologies and the Hodge--Kodaira decompositions
            \begin{equation}
                \sfC^{(i)p}\ \cong\ H^p(\sfC^{(i)})\oplus\underbrace{\im(\sfh^{(i)}_{p+1}\circ\sfd^{(i)}_p)}_{\eqqcolon\,B^p(\sfC^{(i)})}\oplus\underbrace{\im(\sfd^{(i)}_{p-1}\circ\sfh^{(i)}_p)}_{\eqqcolon\,C^p(\sfC^{(i)})}
            \end{equation} 
            for all $p\in\IZ$, as in the proof of \cref{prop:gauge_change}. We also note that
            \begin{equation}
                \sfd^{(i)}\big|_{B^p(\sfC^{(i)})}\,:\,B^p(\sfC^{(i)})\ \overset{\cong}{\rightarrow}\ C^{p+1}(\sfC^{(i)})
                \eand
                \sfh^{(i)}\big|_{C^{p+1}(\sfC^{(i)})}\,:\,C^{p+1}(\sfC^{(i)})\ \overset{\cong}{\rightarrow}\ B^p(\sfC^{(i)})~.
            \end{equation}
            
            We now use \cref{prop:gauge_change} to replace the deformation retract $(\sfh^{(1)},\sfp^{(1)},\sfe^{(1)})$ by $(\tilde\sfh^{(1)},\tilde \sfp^{(1)},\tilde\sfe^{(1)})$ by
            \begin{equation}\label{eq:redefe1}
                \tilde\sfe^{(1)}\ \coloneqq\ \sfe\circ\sfe^{(2)}\circ(\sfe_*)^{-1}~,
            \end{equation} 
            arriving at the commuting diagram
            \begin{equation}\label{eq:basicDiagramExtension}
                \begin{tikzcd}[column sep=50pt]
                    \ar[loop,out=70,in=120,distance=20,"\tilde\sfh^{(1)}" above]\sfC^{(1)}\arrow[d,shift left]{}{\tilde\sfp^{(1)}} & \ar[loop,out=70,in=120,distance=20,"\sfh^{(2)}" above]\sfC^{(2)}\arrow[l,shift left]{}{\sfe}\arrow[d,shift left]{}{\sfp^{(2)}}
                    \\
                    H^{\bullet}(\sfC^{(1)})\arrow[u,shift left]{}{\tilde\sfe^{(1)}} & H^\bullet(\sfC^{(2)})\arrow[u,shift left]{}{\sfe^{(2)}}\arrow[l,"\sfe_*"']
                \end{tikzcd}
            \end{equation}
            where $\sfe_*$ is the map induced by $\sfe$ between the cohomologies. Note that
            \begin{equation}
                \tilde\sfe^{(1)}\circ\sfe_*\circ\sfp^{(2)}\ =\ \sfe\circ\sfe^{(2)}\circ\sfp^{(2)}~,
            \end{equation} 
            and therefore
            \begin{equation}
                \sfe|_{H^\bullet(C^{(2)})}\ =\ \tilde\sfe^{(1)}\circ\sfe_*\circ\sfp^{(2)}~.
            \end{equation} 
            Using further that $\sfe$ is a cochain map, we conclude that its non-trivial components are
            \begin{equation}
				\begin{aligned}
                    &\sfe|^{H^\bullet(\sfC^{(1)})}_{H^\bullet(\sfC^{(2)})}~,~~~
                    &&
                    &&
                    \\
                    &\sfe|^{H^\bullet(\sfC^{(1)})}_{B^\bullet(\sfC^{(2)})}~,~~~
                    &&\sfe|^{B^\bullet(\sfC^{(1)})}_{B^\bullet(\sfC^{(2)})}~,~~~
                    &&\sfe|^{C^\bullet(\sfC^{(1)})}_{B^\bullet(\sfC^{(2)})}~,~~~
                    \\
                    &
                    &&
                    &&\sfe|^{C^\bullet(\sfC^{(1)})}_{C^\bullet(\sfC^{(2)})}~.
				\end{aligned}
			\end{equation}
            Similarly, the projection $\sfp$ we construct will have non-trivial components 
            \begin{equation}
				\begin{aligned}
                    &\sfp|^{H^\bullet(\sfC^{(2)})}_{H^\bullet(\sfC^{(1)})}~,~~~
                    &&
                    &&
                    \\
                    &\sfp|^{H^\bullet(\sfC^{(2)})}_{B^\bullet(\sfC^{(1)})}~,~~~
                    &&\sfp|^{B^\bullet(\sfC^{(2)})}_{B^\bullet(\sfC^{(1)})}~,~~~
                    &&\sfp|^{C^\bullet(\sfC^{(2)})}_{B^\bullet(\sfC^{(1)})}~,~~~
                    \\
                    &
                    &&
                    &&\sfp|^{C^\bullet(\sfC^{(2)})}_{C^\bullet(\sfC^{(1)})}~.
				\end{aligned}
			\end{equation}
            Concretely, we define
            \begin{equation}
				\begin{aligned}
                    \sfp|^{H^\bullet(\sfC^{(2)})}_{H^\bullet(\sfC^{(1)})}\ &\coloneqq\ \left(\sfe|^{H^\bullet(\sfC^{(1)})}_{H^\bullet(\sfC^{(2)})}\right)^{-1}~,
                    \\
                    \sfp|^{C^\bullet(\sfC^{(2)})}_{C^\bullet(\sfC^{(1)})}\ &\coloneqq\ 
                    \sfd^{(2)}\circ\sfp|_{\sfB^\bullet(\sfC^{(1)})}\circ\tilde\sfh^{(1)}~,
				\end{aligned}
			\end{equation}
            where $\sfp|_{\sfB^\bullet(\sfC^{(1)})}$ remains to be fixed. Because $\sfe$ is a cochain map, $\sfe|^{B^\bullet(\sfC^{(1)})}_{B^\bullet(\sfC^{(2)})}$ is an injection, and, hence, we can choose a projection $\sfp|^{B^\bullet(\sfC^{(2)})}_{B^\bullet(\sfC^{(1)})}$ such that
            \begin{equation}
                \sfp|^{B^\bullet(\sfC^{(2)})}_{B^\bullet(\sfC^{(1)})}\circ\sfe|^{B^\bullet(\sfC^{(1)})}_{B^\bullet(\sfC^{(2)})}\ =\ \sfid_{B^\bullet(\sfC^{(2)})}~.
            \end{equation} 
            We then define further
            \begin{equation}
				\begin{aligned}
                    \sfp|^{H^\bullet(\sfC^{(2)})}_{B^\bullet(\sfC^{(1)})}\ &\coloneqq\ 
                    \sfp|^{H^\bullet(\sfC^{(2)})}_{H^\bullet(\sfC^{(1)})}\circ\sfe|^{H^\bullet(\sfC^{(1)})}_{B^\bullet(\sfC^{(2)})}    
                    \circ
                    \sfp|^{B^\bullet(\sfC^{(2)})}_{B^\bullet(\sfC^{(1)})}~,
                    \\
                    \sfp|^{C^\bullet(\sfC^{(2)})}_{B^\bullet(\sfC^{(1)})}\ &\coloneqq\ 
                    \sfp|^{C^\bullet(\sfC^{(2)})}_{C^\bullet(\sfC^{(1)})}\circ\sfe|^{C^\bullet(\sfC^{(1)})}_{B^\bullet(\sfC^{(2)})}
                    \circ
                    \sfp|^{B^\bullet(\sfC^{(2)})}_{B^\bullet(\sfC^{(1)})}~.
				\end{aligned}
			\end{equation}
            It is then straightforward to check that
            \begin{equation}
                \sfp\circ\sfe\ =\ \sfid_{\sfC^{(2)}}~.
            \end{equation} 
            Note that $\sfp$ is indeed a chain map, because
            \begin{equation}
                \sfd^{(2)}\circ\sfp|_{B^\bullet(\sfC^{(1)})}\ =\ \sfd^{(2)}\circ \sfp|_{B^\bullet(\sfC^{(1)})}\circ\tilde\sfh^{(1)}\circ\sfd^{(1)}\ =\ \sfp|_{C^\bullet(\sfC^{(1)})}\circ\sfd^{(1)}~.
            \end{equation} 
            
            It remains to construct a contracting homotopy $\sfh$. Using~\eqref{eq:redefe1}, we have
            \begin{equation}
                \im(\tilde\sfe^{(1)})\ \cong\ H^\bullet(\sfC^{(1)})~.
            \end{equation} 
            Hence, in switching from $\sfe^{(1)}$ to $\tilde\sfe^{(1)}$ using \cref{prop:gauge_change},
            we note that $\psi=0$. It then follows that 
            \begin{equation}
                \tilde\sfp^{(1)}\ \coloneqq\ \sfe_*\circ\sfp^{(2)}\circ\sfp~.
            \end{equation}
            The desired contracting homotopy is now given by 
            \begin{equation}
                \sfh\ \coloneqq\ \tilde\sfh^{(1)}-\sfe\circ\sfh^{(2)}\circ\sfp
            \end{equation} 
            with $\tilde\sfh^{(1)}$ given by \cref{prop:gauge_change}. Indeed, we have
            \begin{equation}
                \sfid-\,\sfe\circ\sfp\ =\ \sfh\circ\sfd^{(1)}+\sfd^{(1)}\circ\sfh~,
            \end{equation} 
            as required.
        \end{proof}
        
        \begin{example}\label{exa:SDR}
            We note that the cochain complexes in \cref{prop:quasiIsoSplitSupermanifolds} are split (with respect to the fermionic parts) and so, it is now always possible to construct a special deformation retract, 
            \begin{equation}
                \begin{tikzcd}
                    \ar[loop,out=170,in=190,distance=20,"\sfh" left]\big(\Omega^{0,\bullet}(E[1]),\bar\partial\big)\arrow[r,shift left]{}{\sfp} & \big(\Omega^{0,\bullet}(M,\bigwedge^\bullet E^*),\bar\partial_{\rm red}\big)\arrow[l,shift left]{}{\sfe}
                \end{tikzcd}
            \end{equation}
            where the contracting homotopy can be glanced from the map that takes a $(0,p)$-form $\omega$ to the coboundary $(0,p-1)$-form $c$ removing the fermionic directions at the cohomological level. In particular,
            \begin{equation}
                \sfe\big(\Omega^{0,\bullet}(M,\mbox{$\bigwedge^\bullet$}E^*)\big)\ \subseteq\ \ker(\sfh)
            \end{equation}
            since for differential forms that do not contain anti-holomorphic fermionic directions, the action of the contracting homotopy vanishes identically.
        \end{example}
        
        \paragraph{Homological perturbation theory.}
        An \uline{$L_\infty$-algebra structure} $\frL^{(1)}$ on a cochain complex $\sfC^{(1)}$ consists of additional products $\mu_i^{(1)}$ of degree $2-i$ with $i=2,3,4,\ldots$ and subject to the \uline{homotopy Jacobi identities}
        \begin{equation}\label{eq:homotopyJacobiIdentities}
            \sum_{j+k=i}\sum_{\sigma\in\overline{\rm Sh}(j;i)}\chi(\sigma;\ell_1,\ldots,\ell_{i})(-1)^{k}\mu^{(1)}_{k+1}\big(\mu^{(1)}_j\big(\ell_{\sigma(1)},\ldots,\ell_{\sigma(j)}\big),\ell_{\sigma(j+1)},\ldots,\ell_{\sigma(i)}\big)\ =\ 0
        \end{equation}
        for all $\ell_1,\ell_2,\ell_3,\ldots\in\frL^{(1)}$. Here, the sum is taken over all $(j;i)$ \uline{unshuffles} $\sigma$ which consist of permutations $\sigma$ of $\{1,\ldots,i\}$ such that the first $j$ and the last $i-j$ images of $\sigma$ are ordered: $\sigma(1)<\cdots<\sigma(j)$ and $\sigma(j+1)<\cdots<\sigma(i)$. Moreover, $\chi(\sigma;\ell_1,\ldots,\ell_i)$ is the \uline{Koszul sign}. 
        
        Now, given a deformation retract~\eqref{eq:deformationRetract}, the \uline{homological perturbation lemma} states that an $L_\infty$-structure $\frL^{(1)}$ on the cochain complex $\sfC^{(1)}$ can be transferred to an equivalent or quasi-isomorphic $L_\infty$-algebra structure $\frL^{(2)}$ on the cochain complex $\sfC^{(2)}$, and the formulas are recursive. More explicitly, the map $\sfe$ in~\eqref{eq:deformationRetract} extends as
        \begin{subequations}
            \begin{equation}\label{eq:homotopTransferQuasiIsomorphism}
                \begin{aligned}
                    \sfT_1(\ell_1)\ &\coloneqq\ \sfe(\ell_1)~,
                    \\
                    \sfT_2(\ell_1,\ell_2)\ &\coloneqq\ -(\sfh\circ\mu^{(1)}_2)(\sfe(\ell_1),\sfe(\ell_2))~,
                    \\
                    &~~\vdots
                    \\
                    \sfT_i(\ell_1,\ldots,\ell_i)\ &\coloneqq\ -\sum_{j=2}^i\frac{1}{j!}\sum_{k_1+\cdots+k_j=i}\sum_{\sigma\in\overline{\rm Sh}(k_1,\ldots,k_{j-1};i)}\chi(\sigma;\ell_1,\ldots,\ell_i)\zeta(\sigma;\ell_1,\ldots,\ell_i)
                    \\
                    &\kern1cm\times(\sfh\circ\mu^{(1)}_j)\big(\sfT_{k_1}\big(\ell_{\sigma(1)},\ldots,\ell_{\sigma(k_1)}\big),\ldots,\sfT_{k_j}\big(\ell_{\sigma(k_1+\cdots+k_{j-1}+1)},\ldots,\ell_{\sigma(i)}\big)\big)
                \end{aligned}
            \end{equation}
            for all $\ell_1,\ell_2,\ell_3,\ldots\in\frL^{(2)}$ and where
            \begin{equation}
                \zeta(\sigma;\ell_1,\ldots,\ell_i)\ \coloneqq\ (-1)^{\sum_{1\leq m<n\leq j}k_mk_n+\sum_{m=1}^{j-1}k_m(j-m)+\sum_{m=2}^j(1-k_m)\sum_{k=1}^{k_1+\cdots+k_{m-1}}|\ell_{\sigma(k)}|}~,
            \end{equation}
            and the higher products $\mu_i^{(2)}$ on $\frL^{(2)}$ are of the form
            \begin{equation}\label{eq:homotopTransferHigherProducts}
                \begin{aligned}
                    \mu^{(2)}_2(\ell_1,\ell_2)\ &\coloneqq\ \sfp(\mu_2^{(1)}(\sfe(\ell_1),\sfe(\ell_2)))~,
                    \\
                    &~~\vdots
                    \\
                    \mu^{(2)}_i(\ell_1,\ldots,\ell_i)\ &\coloneqq\ \sum_{j=2}^i\frac{1}{j!}\sum_{k_1+\cdots+k_j=i}\sum_{\sigma\in\overline{\rm Sh}(k_1,\ldots,k_{j-1};i)}\chi(\sigma;\ell_1,\ldots,\ell_i)\zeta(\sigma;\ell_1,\ldots,\ell_i)
                    \\
                    &\kern1cm\times (\sfp\circ\mu^{(1)}_j)\big(\sfT_{k_1}\big(\ell_{\sigma(1)},\ldots,\ell_{\sigma(k_1)}\big),\ldots,\sfT_{k_j}\big(\ell_{\sigma(k_1+\cdots+k_{j-1}+1)},\ldots,\ell_{\sigma(i)}\big)\big)
                \end{aligned}
            \end{equation}
        \end{subequations}
        for all $\ell_1,\ell_2,\ell_3,\ldots\in\frL^{(2)}$. This makes $\sfT$ into an $L_\infty$-quasi-isomorphism; see~\eqref{eq:LinftyMorphism}. The two important points to notice are that the arguments of the higher products $\mu_i^{(2)}$ are always first mapped to $\frL^{(1)}$ by the embedding $\sfe$, and the $\mu_i^{(2)}$ themselves are then produced by inserting the images of $\sfh$ of lower products into each other. See e.g.~\cite{Jurco:2018sby,Jurco:2019bvp} for details.
        
        \paragraph{Proof of \cref{prop:WittenGauge}.} 
        We have now all the ingredients to prove \cref{prop:WittenGauge}. Firstly, we note that because of the explicit form of the vector fields~\eqref{eq:tangentBundleDistributionsCRSuperAmbitwistorSpaceTwisted} and the commutation relation~\eqref{eq:superspaceTorsionTwisted} and because the CR holomorphic and CR anti-holomorphic fermionic combinations~\eqref{eq:CRFermionicCoordinates} essentially play the role of ordinary holomorphic and anti-holomorphic fermionic coordinates\footnote{See also \cref{rmk:CRAmbitwistorSpace} for an alternative reason for why the fermionic combinations~\eqref{eq:CRFermionicCoordinates} should be regarded as holomorphic coordinates.} (that is, $F$ is essentially a globally split CR supermanifold), the proof of \cref{prop:quasiIsoSplitSupermanifolds} goes through also for our twisted CR differentials~\eqref{eq:CRantiholomorphicDifferentialTwisted} and~\eqref{eq:CRantiholomorphicDifferentialTwistedReduced}. 
        
        Next, by virtue of this discussion and as explained in \cref{exa:SDR}, we now also have a special deformation retract 
        \begin{equation}
            \begin{tikzcd}
                \ar[loop,out=170,in=190,distance=20,"\sfh" left]\underbrace{\big(\Omega^{0,\bullet}_{\rm CR,\,tw}(F),\bar\partial_{\rm CR,\,tw}\big)}_{=\,\sfC^{(1)}}\arrow[r,shift left]{}{\sfp} & \underbrace{\big(\Omega^{0,\bullet}_{\rm CR,\,tw,\,red}(F),\bar\partial_{\rm CR,\,tw,\,red}\big)}_{=\,\sfC^{(2)}}\arrow[l,shift left]{}{\sfe}
            \end{tikzcd}
        \end{equation}
        with $\sfe\big(\Omega^{0,\bullet}_{\rm CR,\,tw,\,red}(F)\big)\subseteq\ker(\sfh)$. We can therefore apply the homological perturbation lemma for the natural product $\mu_2^{(1)}(-,-)=[-,-]$. In particular, the formulas~\eqref{eq:homotopTransferHigherProducts} now show that the binary product is the expected one, $\mu_2^{(2)}(-,-)=[-,-]$. Finally, we have $(\sfh\circ\mu_2^{(2)})(\sfe(-),\sfe(-))=0$, since the wedge product of forms without fermionic directions is a form without fermionic directions. This renders all higher products $\mu_i^{(2)}$ with $i>2$ in~\eqref{eq:homotopTransferHigherProducts} trivial, and we arrive at the desired result.
        
        \begin{remark}\label{rem:inf-dim}
            We heavily use the Hodge--Kodaira decomposition, splitting maps, and aspects of cyclic structures, some of which tend to become problematic for infinite-dimensional vector spaces. Since our $L_\infty$-algebras consist of spaces of differential forms which are infinite-dimensional, so let us briefly comment on this point and explain that these infinite-dimensionalities are harmless for the purpose of most perturbative field theories. 
            
            As is common in the literature, we demand that we have a factorisation
            \begin{equation}
                \frL_p\ \cong\ \frV_p\otimes C^\infty_p
            \end{equation}
            of the homogeneously graded subspaces $\frL_p$ of any $L_\infty$-algebra $\frL$ we employ, where $p\in\IZ$ and $\frV_p$ is a finite-dimensional vector space. The pure function spaces $C^\infty_p$ are then constructed, also as common in the literature on perturbative field theory, as a finite linear combination of some preferred basis, e.g.~plane waves on Minkowski space, or spherical harmonics on $\IC P^1$ and so, the differential cochain complex of the corresponding $L_\infty$-algebras can be consistently truncated to finite-dimensional $C^\infty_p$. 
            
            As an example, consider the differential cochain complex underlying $\frL_{\rm YM_1}$ as defined in~\eqref{eq:DGLAFirstOrderYM}. It is clear that $C_p^\infty$ here can be consistently truncated to individual plane waves of a specific momentum, and the cochain complex $\frL_{\rm YM_1}$ splits into a direct sum of finite-dimensional cochain complexes.
            
            For the differential cochain complex underlying $\frL_{\rm CR,\,tw,\,red}$, we can use plane waves on $\IR^4$ together with spherical harmonics on $\IC P^1\times\IC P^1$. We note that $\bar{\partial}_{\rm CR,\,tw,\,red}$ can maximally increase the angular momentum $\ell$ of the spherical harmonics $Y_{\ell m}$ on the two $\IC P^1$ by one. We can thus consistently restrict $C_p^\infty$ to angular momenta $\ell\leq\ell_0+p$ for some $\ell_0\in\IN$ on both $\IC P^1$, ending up with a finite-dimensional vector space.
        \end{remark}
        
    \end{body}
    
\end{document}